\documentclass[11pt, twoside]{article}
\usepackage{amsmath}
\usepackage{amssymb}			
\usepackage{amsthm}
\usepackage{amsfonts}
\usepackage{mathtools}
\usepackage{empheq}
\usepackage[mathscr]{eucal}
\usepackage{tensor}
\usepackage{tikz}
\usetikzlibrary{tikzmark,fit}							

\usepackage{lmodern} 											
\usepackage[T1]{fontenc} 											
\usepackage{microtype} 											
\usepackage[lmargin=20mm, hmarginratio=1:1, top=25mm, bottom=25mm, columnsep=20pt]{geometry} 		
\usepackage{multicol} 											

\usepackage{abstract} 									
\newtheoremstyle{def}{1.5\topsep}{7pt\topsep}{\upshape}{}{\bfseries}{.}{5pt}	{\thmname{#1}\thmnumber{ #2}}

\newtheoremstyle{obs}{1.5\topsep}{7pt\topsep}{\upshape}{}{\bfseries}{.}{5pt}	{\thmname{#1}\thmnumber{ #2}}

\newtheoremstyle{prp}{1.5\topsep}{7pt\topsep}{\upshape}{}{\bfseries}{.}{5pt}	{\thmname{#1}\thmnumber{ #2}}	

\newtheoremstyle{th}{1.5\topsep}{7pt\topsep}{\upshape}{}{\bfseries}{.}{5pt}	{\thmname{#1}\thmnumber{ #2}}
	
\newtheoremstyle{lmm}{1.5\topsep}{7pt\topsep}{\upshape}{}{\bfseries}{.}{5pt}	{\thmname{#1}\thmnumber{ #2}}

\theoremstyle{th}
\newtheorem{theorem}{Theorem}[section]	

\theoremstyle{prp}
\newtheorem{proposition}{Proposition}[section]

\theoremstyle{def}
\newtheorem{definition}{Definition}[section]

\theoremstyle{obs}

\theoremstyle{lmm}
\newtheorem{lemma}{Lemma}[section]

\usepackage{fancyhdr} 			

\usepackage[titletoc]{appendix}			

\usepackage{cite}
\usepackage{paralist} 										
\usepackage[hang, small,labelfont=bf,up,textfont=it,up]{caption} 		
\usepackage{booktabs} 										
\usepackage{float} 					
\usepackage{cases}
\usepackage{multicol}

\usepackage{slashed}

\usepackage{hyperref} 				

\fancypagestyle{plain}{%
	\fancyhead[R]{IFT-UAM/CSIC-19-97}
	\renewcommand{\headrulewidth}{0pt}
}

\title{\vspace{15mm}\fontsize{14pt}{15pt}\selectfont\textbf{$p$-brane Newton--Cartan Geometry}\vspace{5mm}} 

\author{
	\large
	\textsc{David Pere\~niguez}\footnote{david.perenniguez@uam.es}\\[12mm] 			
	\textsc{\textit{Instituto de Física Teórica UAM/CSIC}}\\
	\textsc{ \textit{C/ Nicolás Cabrera, 13-15, C.U. Cantoblanco, E-28049 Madrid, Spain}}
}

\begin{document}				

\begin{titlepage} 			
\centering

\maketitle 						
\fancyfoot{}


\begin{abstract}
\noindent We provide a formal definition of $p$-brane Newton--Cartan ($p$NC) geometry and establish some foundational results. Our approach is the same followed in the literature for foundations of Newton--Cartan Gravity. Our results provide control of aspects of $p$NC geometry that are otherwise unclear when using the usual gauge language of non-relativistic theories of gravity. In particular, we obtain a set of necessary and sufficient conditions that a $p$NC structure must satisfy in order to admit torsion-free, compatible affine connections, and determine the space formed by the latter. This is summarised in Theorem \ref{th:1}. Since $p$NC structures interpolate between Leibnizian structures for $p=0$ and Lorentzian structures for $p=d-1$ (with $d$ the dimension of the spacetime manifold), the present work also constitutes a generalisation of results of Newton--Cartan and (pseudo-) Riemannian geometry.
\end{abstract}

\vfill


\end{titlepage}


\pagestyle{empty}
\newpage
\tableofcontents

\newpage


\setcounter{page}{1}
\pagestyle{fancy}
\fancyhead{}
\fancyfoot{}
\fancyfoot[C]{-- ~\thepage~ --} 		

\renewcommand{\headrulewidth}{0pt}

\section{Introduction}
\subsection{Motivation}

From the celebrated work of É. Cartan in \cite{E.Cartan}, we learnt that the geometric description of gravity is something by no means unique of General Relativity. In particular, he showed that Newton's theory of gravity can also be reformulated in purely geometrical terms, giving raise to the so-called Newton--Cartan Gravity. The geometry required to do so, however, is not (pseudo-) Riemannian, and much work has been done to construct the foundations of such theory (see, e.g.\,\cite{XBKM,Malament}). More recently, it has been discovered that there are more non-relativistic theories of gravity beyond Newton's. For instance, it is possible to define a type of non-Riemannian geometry, called Stringy Newton--Cartan geometry (SNC), that describes the gravitational field that couples to non-relativistic strings. More precisely, the non-relativistic closed string theory described in \cite{NRCS} can be consistently coupled to a SNC background geometry. This was first noticed in \cite{SNC}, and recent work developing applications of such result can be found in \cite{NRSTTD,Tachyons}. Newton--Cartan Gravity has also been extended by considering the addition of torsional connections. Such geometries, dubbed Torsional Newton--Cartan (TNC), posses interesting relations with holography \cite{ObersHolography1,ObersHolography2}, Horava--Lifshitz Gravity \cite{HL} and non-relativistic string theory \cite{ObersST1,ObersST2,ObersST3}.

SNC geometry was first introduced using the modern language in which non-relativistic theories are constructed. Essentially, such formalism consists in gauging a given non-relativistic Lie algebra and, by imposing a set of \textit{conventional} curvature constraints, the field content is reduced. This results in a geometry that encodes the degrees of freedom of the theory (see \cite{Bargmann} for an illustrative example). Alternative methods for deriving such non-relativistic geometries have also been developed \cite{GalileiFieldTheory}. The construction of invariant actions for such theories is a difficult task, although some success has been achieved for particular geometries such as (type II) TNC gravity \cite{ObersAction1,ObersAction2}. More systematic procedures towards the construction of actions have been studied in \cite{LucaTomas}. This approach to non-relativistic gravity has proven to be very powerful and the resulting theories have potentially interesting applications in holography and condensed matter physics \cite{Hoyos,Son,ObersHolography1,ObersHolography2}. However, as we will discuss in the next section, such formalism, although physically useful, obscures some aspects of the resulting geometries. These are better understood when reformulated in a framework similar to that used in the literature for foundations of Newton--Cartan Gravity (e.g.\,\cite{XBKM,Malament}). Such formalism is briefly reviewed in the next section.

In this paper we consider the extension of SNC to the case of $p$-branes and provide precise control of some aspects of the geometry that have not yet been investigated in full detail.
\subsection{Non-Relativistic Geometries}

Following É.Cartan's work, the geometric description of gravity is based on two structures:
\begin{enumerate}
\item First, a tensor structure must be prescribed on the spacetime manifold $M$ in order to realise a given structure group. For example, the Lorentz group $O(1,d-1)$ (where $d=\text{dim}M$) can be realised on $M$ by prescribing a metric tensor of Lorentzian signature, $g_{ab}$. Indeed, the bundle of frames on $M$ that are orthonormal with respect to $g_{ab}$ form a reduction of the frame bundle based on $O(1,d-1)$.

\item Second, a notion of spacetime curvature is provided by prescribing a connection on the tangent bundle $TM$. In addition, such connection is required to be compatible with the tensor structure that realises the structure group (some alternative compatibility conditions for the case of $p$-brane Newton--Cartan geometry are discussed in \cite{TomasDavidLuca1}). Considering the previous example where the tensor structure consists of a Lorentzian metric, the Fundamental Theorem of Riemannian Geometry guarantees the existence and uniqueness of a torsion-free connection on $TM$ which is compatible with $g_{ab}$. This is, of course, the Levi--Civita connection of $g_{ab}$.
\end{enumerate}
Non-relativistic geometries can be constructed by first choosing a structure group consisting of some 'non-relativistic limit' of $O(1,d-1)$. A canonical example of non-relativistic group is the homogeneous Galilei group $\text{Gal}(d)$. The tensor structure that realises it is a pair $(\tau_{a},h^{ab})$ where $\tau_{a}$ is a no-where vanishing 1-form (the \textit{absolute clock}), $h^{ab}$ is a symmetric tensor of Riemannian signature and rank $d-1$ (the \textit{absolute rulers}), and both are mutually orthogonal $\tau_{a}h^{ab}=0$. The pair $(\tau_{a},h^{ab})$ is said to form a \textit{Leibnizian structure on $M$} \cite{XBKM,BernalSanchez}. In general, the tensor structures that realise non-relativistic groups do \textit{not} contain a (pseudo-) Riemannian metric. Thus, there is no analogue of the Levi--Civita connection. More precisely, a connection that is torsion-free and compatible with the tensor structure might be non-unique, or might not exist. In the mathematical literature, this is sometimes referred to as the \textit{equivalence problem}. Solving it consists in determining which additional structure must be prescribed on the manifold, e.g.\ torsion, in order to fix a connection uniquely.

The equivalence problem is well understood for Leibnizian structures. First, it has been determined which subclass of structures admit torsion-free, compatible connections. These are defined by the property of closedness of the absolute clock, $\text{d}\tau=0$, and are dubbed \textit{Augustinian structures}. Then, for such structures it can be shown that, given a \textit{field of observers} $N$ (a vector field satisfying $\tau(N)=1$ everywhere), a torsion-free connection compatible with $(\tau_{a},h^{ab})$ can be uniquely determined. It is referred to as the \textit{torsion-free special connection associated to N} \cite{Torfreespecial1,Torfreespecial2}. There is one of those for each $N$, and they can be thought of as the analogues of the Levi--Civita connection.

However, the equivalence problem associated to SNC structures (and their generalisation to $p$-branes) has not been studied in detail. Thus, some aspects of these geometries remain unclear. For instance, it has not been determined whether the set of conventional curvature constraints in \cite{SNC} are necessary for ensuring existence of torsion-free connections that are compatible with the SNC metrics. It is worth having precise control of this, as one might be interested in relaxing the curvature constraints of the theory without losing the notion of torsion-free, compatible connections. Similarly, the space of all such connections has not been determined precisely. We will see that some symmetries of the compatible connections were originally missed and, hence, the space in which the latter live is actually smaller than initially suggested. Exact determination of such space is convenient since, in general, the fields parametrising the connections propagate degrees of freedom (as it is the case in the Newton--Cartan theory \cite{Malament}, for instance). 

In the present paper we study the equivalence problem associated to \textit{p-brane Newton--Cartan} structures ($p$NC). To this aim, the first part of the work is devoted to the reformulation of $p$NC geometry in a framework that generalises the one used to establish the foundations of Newton--Cartan Gravity \cite{XBKM,Malament}. This language, closer to the initial idea of É.Cartan presented above, is more convenient for the study of the equivalence problem. The second part of the paper is dedicated to providing a set of necessary and sufficient conditions that a $p$NC structure must satisfy in order to admit torsion-free, compatible affine connections, and to determine the space formed by the latter. This is summarised in Theorem \ref{th:1}. Our paper fixes some aspects of $p$NC geometry that have not been studied in detail before and, thus, complements previous work in \cite{SNC}. In addition, since $p$NC structures interpolate between Leibnizian structures for $p=0$ and Lorentzian structures for $p=d-1$, our results constitute a generalisation of those known for Newton--Cartan and (pseudo-) Riemannian geometry.

\paragraph{Outline.} In Section 2, we motivate a choice of non-relativistic structure group, dubbed $G_{p}$, by studying a non-relativistic limit of the worldvolume action of a $p$-brane propagating in a Minkowski background in $d$ dimensions. Then, we derive a tensor structure that realises $G_{p}$ on the spacetime manifold. This gives raise to the notion of $p$-brane Newton--Cartan structures. In Section 3, we first classify $p$NC structures in Aristotelian and Augustinian (analogously to the classification of Leibnizian structures), the latter being the only ones admitting compatible connections with vanishing torsion. Focusing on the Augustinian case, we solve the equivalence problem and determine the corresponding space of torsion-free, compatible connections. Finally, in Section 4, we summarise our results and compare them with previous work in the literature.

\subsection{Conventions}
Lower-case Latin characters in the beginning of the alphabet, $a,b,c,...$, are used as abstract indices while Greek symbols $\alpha,\beta,\gamma...$ are reserved for labelling the components of tensors in (either general or particular) coordinate charts. Capital Latin characters $A,B,C,...$ are used as labels that run from $0$ to $p$, while lower case Latin characters in the middle of the alphabet $i,j,k,...$ are labels running from $1$ to $d-p-1$.

The symbol $\eta_{AB}$ is reserved for the components of the Lorentz metric in $p+1$ dimensions, $\text{diag}(-1,1,...,1)$. $\eta^{AB}$ denotes its inverse. In order to avoid confusion, operations of raising or lowering Latin indices are \textit{not} considered throughout the paper. 

The symmetrisation and anti-symmetrisation operations are defined as
\begin{equation}\label{eq:symasym}
A_{(ab)}:=\frac{1}{2!}(A_{ab}+A_{ba}),\,\,\,\,\,\,\,\,B_{[ab]}:=\frac{1}{2!}(B_{ab}-B_{ba}),
\end{equation}
and are generalised to tensors of arbitrary rank in the obvious way. If $p_{a_{1}...a_{p}},q_{b_{1}...b_{q}}$ are a $p$-form and a $q$-form and $P_{a_{1}...a_{p}},Q_{b_{1}...b_{q}}$ are totally symmetric tensors, then the $\wedge$-product and $\vee$-product are defined as
\begin{equation}\label{wedgevee}
\left(p\wedge q\right)_{a_{1}...a_{p}b_{1}...b_{q}}:=\frac{(p+q)!}{p!q!}p_{[a_{1}...a_{p}}q_{b_{1}...b_{q}]},
\,\,\,\,\,\,
\left(P\vee Q\right)_{a_{1}...a_{p}b_{1}...b_{q}}:=\frac{(p+q)!}{p!q!}P_{(a_{1}...a_{p}}Q_{b_{1}...b_{q})}.
\end{equation}

Finally, if $\nabla$ is an affine connection on the tangent bundle $TM$ and $\{\partial_{\mu}\}$ the frame associated to a general coordinate system, the connection components ${\Gamma^{\lambda}}_{\mu\nu}$ of $\nabla$ are defined, in such chart, as
\begin{equation}\label{concomp}
{\Gamma^{\lambda}}_{\mu\nu}\partial_{\lambda}:=\nabla_{\partial_{\nu}}\partial_{\mu},
\end{equation}
and the torsion tensor $T$ of $\nabla$ is given by
\begin{equation}\label{torsion}
T(X,Y):=\nabla_{X}Y-\nabla_{Y}X-[X,Y],
\end{equation}
where $X,Y$ are tangent vector fields.
\section{$p$-brane Newton--Cartan Structures}

This section is divided in two parts. First, we motivate our choice of structure group, dubbed here $G_{p}$. Then, working on a $d$-dimensional vector space, we construct a tensor structure which defines a class of frames that form a $G_{p}$-torsor\footnote{We recall that a $G$-torsor of a group $G$ is a set on which $G$ acts regularly (i.e.\,freely and transitively).}. These are the analogues of the orthonormal frames in Lorentzian structures. In the second part, we extend the tensor structure to the spacetime manifold, leading to the notion of $p$-brane Newton--Cartan structure. The corresponding principal $G_{p}$-bundle of frames is defined. We also construct other principal bundles, based on subgroups of $G_{p}$, which become crucial in the study of $p$NC connections. 
\subsection{The Group of Symmetries of Non-Relativistic $p$-branes}
Our starting point is the reparametrisation-invariant worldvolume action for a relativistic $p$-brane embedded in Minkowski space in $d$ dimensions,
\begin{equation}\label{eq:21}
S[X^{\mu}]=-T_{p}\int d^{p+1}\sigma \sqrt{-\gamma},
\end{equation}
where $\gamma$ is the pull-back of the Minkowski metric on the worldvolume. We will use $\sigma^{\bar{A}}$ for generic worldvolume coordinates, so that this pull-back is given by
\begin{equation}\label{eq:22}
\gamma_{\bar{A}\bar{B}}=\eta_{\mu\nu}\partial_{\bar{A}}X^{\mu}\partial_{\bar{B}}X^{\nu}.
\end{equation}
For this discussion it is convenient to introduce a particular choice of worldvolume coordinates corresponding to the first $p+1$ target space coordinates $x^{A}$, so that $X^{A}=x^{A}$ and $X^{i}=X^{i}(x^{A})$. Now, following \cite{NonRellim1}, we rescale the longitudinal spacetime coordinates as
\begin{equation}\label{eq:23}
x^{A}\longrightarrow cx^{A}
\end{equation}
with $c\gg 1$. This has the effect of focusing on a small region of the brane. Now $\gamma_{\bar{A}\bar{B}}$ reads
\begin{equation}\label{eq:24}
\gamma_{\bar{A}\bar{B}}=c^{2}\bar{\gamma}_{\bar{A}\bar{B}}+\delta_{ij}\partial_{\bar{A}}X^{i}\partial_{\bar{B}}X^{j}
\end{equation}
where we have introduced the auxiliary worldvolume metric 
\begin{equation}\label{eq:25}
\bar{\gamma}_{\bar{A}\bar{B}}=\eta_{AB}\partial_{\bar{A}}X^{A}\partial_{\bar{B}}X^{B}.
\end{equation}
The inverse of $\bar{\gamma}_{\bar{A}\bar{B}}$ is
\begin{equation}\label{eq:26}
\bar{\gamma}^{\bar{A}\bar{B}}=\eta^{AB}\partial_{A}\sigma^{\bar{A}}\partial_{B}\sigma^{\bar{B}},
\end{equation}
and expanding the determinant of $\gamma$ gives
\begin{equation}\label{eq:27}
S[X^{\mu}]=-T c^{2}\int d^{p+1}\sigma \sqrt{-\bar{\gamma}}\left(1+\frac{1}{2c^{2}}\bar{\gamma}^{\bar{A}\bar{B}}\partial_{\bar{A}}X^{i}\partial_{\bar{B}}X^{j}\delta_{ij}\right)+...
\end{equation}
where we have rescaled the tension as
\begin{equation}\label{eq:28}
T=T_{p}c^{p-1}
\end{equation}
and the ellipsis corresponds to all terms that vanish for $c\to\infty$. The divergent term $c^{2}\int d^{p+1}\sigma\sqrt{-\bar{\gamma}}$ does not affect the dynamics. Indeed, choosing $x^{A}$ as the worldvolume coordinates, this term is $-Tc^{2}\int d^{p+1}x$. Alternatively, this divergent term can be canceled by, for instance, coupling the brane to non-dynamical background fields as described in \cite{SNC,NRCS}. From any perspective, we shall not consider this term here as we are mainly interested in the dynamics of the non-relativistic $p$-brane. The exact limit $c\to\infty$ results in the action
\begin{equation}\label{eq:29}
S[X^{\mu}]=-\frac{T}{2}\int d^{p+1}\sigma \sqrt{-\bar{\gamma}}\left(\bar{\gamma}^{\bar{A}\bar{B}}\partial_{\bar{A}}X^{i}\partial_{\bar{B}}X^{j}\delta_{ij}\right).
\end{equation}
This action is invariant under the spacetime coordinate transformations
\begin{equation}\label{eq:30}
(x^{A},x^{i})\mapsto({\Lambda^{A}}_{B}x^{B}+\xi^{A},{R^{i}}_{j}x^{j}+{v^{i}}_{B}x^{B}+\chi^{i})
\end{equation}
where $R\in O\left(d-p-1\right)$, $\Lambda\in O(1,p)$, ${v^{i}}_{A}=({v^{i}}_{0},...,{v^{i}}_{p})\in \mathbb{R}^{(p+1)(d-p-1)}$, $\xi^{A}\in \mathbb{R}^{p+1}$ and $\chi^{i}\in \mathbb{R}^{d-p-1}$, all of them being constant, finite parameters. 

In order to define the structure group of $p$NC geometry, we consider flat space containing a non-relativistic $p$-brane described by \eqref{eq:29}. The symmetry group of this background is not $ISO(d)$ anymore but the set of transformations in \eqref{eq:30}. In particular, we will need the group of transformations that the spacetime coordinate transformations \eqref{eq:30} induce in the space of frames. We call such a group $G_{p}$. The frames of two coordinate systems connected by a symmetry transformation \eqref{eq:30} are related by
\begin{equation}\label{eq:31}
\left(\partial_{A}',\partial_{j}'\right)=\left(\partial_{B},\partial_{i}\right){\left(\begin{array}{@{}c|c@{}}{\Lambda^{B}}_{A} & 0  \\\hline{v^{i}}_{A} & {R^{i}}_{j}  \\ \end{array}\right)}.
\end{equation}
Hence, $G_{p}$ can be defined as the set of matrices
\begin{equation}\label{eq:32}
G_{p}:=\left\{{\left(\begin{array}{@{}c|c@{}}{\Lambda^{A}}_{B} & 0  \\\hline{v^{i}}_{B} & {R^{i}}_{j}  \\ \end{array}\right)} \,\,\,\,\,\text{with}\,\,\,\,\,  \Lambda\in O(1,p),\,\,\,\,\, R\in O\left(d-p-1\right),\,\,\,\,\, {v^{i}}_{A}\in \mathbb{R}^{(p+1)(d-p-1)} \right\},
\end{equation}
and we notice that the inverse of a generic element $({v^{i}}_{B},{\Lambda^{A}}_{B},{R^{i}}_{j})\in G_{p}$ is given by
\begin{equation}\label{eq:33}
{\left(\begin{array}{@{}c|c@{}}{\left(\Lambda^{-1}\right)^{B}}_{C} & 0  \\\hline -{\left(R^{-1}\right)^{j}}_{l}{v^{l}}_{D}{\left(\Lambda^{-1}\right)^{D}}_{C} & {\left(R^{-1}\right)^{j}}_{k}  \\ \end{array}\right)}.
\end{equation}

$G_{p}$ is the semi-direct product of two smaller groups of matrices. First, we notice it has a normal subgroup, that will be referred to as the longitudinal group $LG_{p}$, given by
\begin{equation}\label{eq:34}
LG_{p}:=\left\{{\left(\begin{array}{@{}c|c@{}}{\Lambda^{A}}_{B} & 0  \\\hline{v^{i}}_{B} & {\delta^{i}}_{j}  \\ \end{array}\right)} \,\,\,\,\,\text{with}\,\,\,\,\,  \Lambda\in O(1,p),\,\,\,\,\, {v^{i}}_{A}\in \mathbb{R}^{(p+1)(d-p-1)} \right\}.
\end{equation}
This allows us to write $G_{p}$ as the semi-direct product\footnote{In order to alleviate the notation we will make no notational difference between inner or outer semi-direct products, neither will we write explicitly the corresponding group homeomorphisms, as all these should be clear by the context.}
\begin{equation}\label{eq:35}
G_{p}=LG_{p}\rtimes O(d-p-1).
\end{equation}
In its turn, $LG_{p}$ has a normal subgroup consisting of the set of matrices of the form 
\begin{equation}\label{eq:36}
\left(\begin{array}{@{}c|c@{}}{\delta^{A}}_{B} & 0  \\\hline{v^{i}}_{B} & {\delta^{i}}_{j}  \\ \end{array}\right),
\end{equation}
so it can be decomposed as
\begin{equation}\label{eq:37}
LG_{p}=\mathbb{R}^{(p+1)(d-p-1)}\rtimes O(1,p).
\end{equation}
Then, we can write $G_{p}$ in the form
\begin{equation}\label{eq:38}
G_{p}=\left( \mathbb{R}^{(p+1)(d-p-1)}\rtimes O(1,p) \right)\rtimes O(d-p-1).
\end{equation}

\subsubsection{The Metric Tensors}
In GR, a smooth Lorentzian metric provides a notion of orthonormal frames. These form a principal $O(1,d-1)$-bundle. In particular, this means that at each spacetime point, the set of orthonormal frames are a $O(1,d-1)$-torsor. The purpose of this section is to obtain a set of tensors that realise the group $G_{p}$ in the same sense a Lorentzian metric realises $O(1,d-1)$. These tensors are first obtained on a $d$-dimensional vector space and then extended to the spacetime manifold.

Let us denote $\mathcal{V}$ a real vector space of dimension $d$ and $F(\mathcal{V})$ the corresponding space of frames. For convenience, we shall use the notation $(\tau_{A},e_{i})$ for the elements in $F(\mathcal{V})$, and $(\tau^{A},e^{i})$ for those in $F^{*}(\mathcal{V})$. The group $G_{p}$ acts on $F(\mathcal{V})$ from the right as\footnote{For clarity, let us remark that by this notation we mean $ (\tau_{A},e_{i}){\left(\begin{array}{@{}c|c@{}}{\Lambda^{A}}_{B} & 0  \\\hline{v^{i}}_{B} & {R^{i}}_{j}  \\ \end{array}\right)}=\left({\Lambda^{A}}_{B}\tau_{A}+{v^{i}}_{B}{e}_{i}, {R^{i}}_{j}e_{i}\right)$}
\begin{align}\label{eq:2121}
F(\mathcal{V})\times G_{p} &\longrightarrow F(\mathcal{V})\\
\left((\tau_{A},e_{i}),g\right) & \mapsto (\tau_{B},e_{j})\cdot g= (\tau_{A},e_{i}){\left(\begin{array}{@{}c|c@{}}{\Lambda^{A}}_{B} & 0  \\\hline{v^{i}}_{B} & {R^{i}}_{j}  \\ \end{array}\right)}. \notag
\end{align}
Notice that, if $(\tau_{A},e_{i}),(\tau'_{B},e'_{j})\in F(\mathcal{V})$ are related by \eqref{eq:2121}, then the associated dual frames $(\tau^{A},e^{i}),(\tau'^{B},e'^{j})\in F^{*}(\mathcal{V})$ transform as
\begin{equation}\label{eq:2122}
{\left(\begin{array}{c}\tau'^{A}  \\ e'^{i} \end{array}\right)}={\left(\begin{array}{@{}c|c@{}}{\left(\Lambda^{-1}\right)^{A}}_{B} & 0  \\\hline-{\left(R^{-1}\right)^{i}}_{k}{v^{k}}_{C}{\left(\Lambda^{-1}\right)^{C}}_{B} & {\left(R^{-1}\right)^{i}}_{j}  \\ \end{array}\right)}{\left(\begin{array}{c}\tau^{B}  \\ e^{j} \end{array}\right)}.
\end{equation}
The tensor equalities
\begin{equation}\label{eq:2123}
\eta_{AB}\tau^{A}\otimes\tau^{B}=\eta_{AB}\tau'^{A}\otimes\tau'^{B}, \,\,\,\,\,\,\,\,\, \delta^{ij}e_{i}\otimes e_{j}=\delta^{ij}e'_{i}\otimes e'_{j},
\end{equation}
are manifest. Thus, taking $(\tau_{A},e_{i})\in F(\mathcal{V})$ and defining $\tau:=\eta_{AB}\tau^{A}\otimes\tau^{B}$ and $h:=\delta^{ij}e_{i}\otimes e_{j}$, it is clear that the orbit of $G_{p}$ through $(\tau_{A},e_{i})$ consists of bases orthonormal with respect to $\tau$ and $h$. Although less manifest, there is a third tensor that can be constructed and that is invariant under the action of $G_{p}$. Defining $\tau$ as before, a Riemannian metric can be defined on $\text{Ker}(\tau)$\footnote{We take $\text{Ker}(\tau)=\{v\in\mathcal{V} \,\, \vert \,\, \tau(v,\cdot)=0 \}$, and notice $\text{dim}\left(\text{Ker}(\tau)\right)=d-(p+1)$.} as $\gamma:=\left(\delta_{ij}e^{i}\otimes e^{j}\right)\rvert_{\text{Ker}(\tau)}$. Again, from \eqref{eq:2121} it follows that the orbit of $G_{p}$ through $(\tau_{A},e_{i})$ is formed by bases orthonormal with respect to $\tau$ and $\gamma$.

Reversing the perspective, it is sensible to expect that for a given pair of tensors $(\tau,\gamma)$ the corresponding set of orthonormal frames is precisely the desired $G_{p}$-torsor. This is formally expressed in the following definition and proposition.
\begin{definition}\label{def:2121}
\textit{Let $\tau\in \vee^{2}\mathcal{V}^{*}$ with $\text{rank}(\tau)=p+1$ and Lorentzian signature, and let $\gamma$ be a metric of Riemannian signature in $\text{Ker}(\tau)$. We define the space of Galilean $p$-frames, denoted $f_{p}(\mathcal{V},\tau,\gamma)$, as the space of frames that are orthonormal with respect to $\tau$ and $\gamma$, that is}
\begin{equation}\label{eq:2123A}
f_{p}(\mathcal{V},\tau,\gamma):=\left\{(\tau_{A},e_{i})\in F(\mathcal{V})\,\,\vert \{e_{i}\}\in\text{Ker}(\tau), \,\,\,\,\tau(\tau_{A},\tau_{B})=\eta_{AB}, \,\,\,\, \gamma(e_{i},e_{j})=\delta_{ij} \right\}.
\end{equation}
\end{definition}

\begin{proposition}\label{prp:2121}
\textit{The space $f_{p}(\mathcal{V},\tau,\gamma)$ is a $G_{p}$-torsor with respect to the (right) action}
\begin{align}\label{eq:2125}
f_{p}(\mathcal{V},\tau,\gamma)\times G_{p} &\longrightarrow f_{p}(\mathcal{V},\tau,\gamma),\\
(\left(\tau_{A},e_{i}),g\right) & \mapsto (\tau_{B},e_{j})\cdot g=(\tau_{A},e_{i}){\left(\begin{array}{@{}c|c@{}}{\Lambda^{A}}_{B} & 0  \\\hline{v^{i}}_{B} & {R^{i}}_{j}  \\ \end{array}\right)}. \notag
\end{align}

\begin{proof}
Regularity of the action can be proven by showing that for each $(\tau_{B},e_{j})\in f_{p}(\mathcal{V},\tau,\gamma)$ the map $G_{p}\to f_{p}(\mathcal{V},\tau,\gamma)$ given by $g\mapsto (\tau_{B},e_{j})\cdot g$ is a bijection. That it is injective follows immediately by construction. To show that it is surjective take any $(\tau'_{B},e'_{j})\in f_{p}(\mathcal{V},\tau,\gamma)$. Since both $(\tau'_{B},e'_{j})$ and $(\tau_{B},e_{j})$ belong to $F(\mathcal{V})$ they must be related by
\begin{equation}\label{eq:2127}
(\tau'_{B},e'_{j})=(\tau_{B},e_{j})M,
\end{equation}
where $M\in GL(d)$. Decomposing $M$ suggestively as
\begin{equation}\label{eq:2128}
M={\left(\begin{array}{@{}c|c@{}}{D^{A}}_{B} & {u^{A}}_{j} \\\hline{v^{i}}_{B} & {F^{i}}_{j}  \\ \end{array}\right)},
\end{equation}
one has
\begin{equation}\label{eq:2129}
(\tau'_{B},e'_{j})=({D^{A}}_{B}\tau_{A}+{v^{i}}_{B}e_{i},{u^{A}}_{j}\tau_{A}+{F^{i}}_{j}e_{i}).
\end{equation}
Now by imposing that both frames are in $f_{p}(\mathcal{V},\tau,\gamma)$ one gets
\begin{align}\label{eq:21210}
&i)\,\,\,\,\, \eta_{BC}=\tau(\tau'_{B},\tau'_{C})={D^{A}}_{B}{D^{E}}_{C}\tau(\tau_{A},\tau_{E})={D^{A}}_{B}{D^{E}}_{C}\eta_{AE}\rightarrow {D^{A}}_{B}\in O(1,p),\\ \notag
&ii)\,\,\,\,\,0=\tau(e'_{j},\cdot)={u^{A}}_{j}\tau(\tau_{A},\cdot)\rightarrow 0={u^{A}}_{j}\eta_{AB}\rightarrow {u^{A}}_{j}=0,\\
&iii)\,\,\,\,\, \delta_{ij}=\gamma(e'_{i},e'_{j})={F^{k}}_{i}{F^{l}}_{j}\gamma(e_{k},e_{l})={F^{k}}_{i}{F^{l}}_{j}\delta_{kl}\rightarrow {F^{i}}_{j}\in O(d-(p+1)) \notag.
\end{align}
That is, $M\in G_{p}$ and consequently the map $G_{p}\to f_{p}(\mathcal{V},\tau,\gamma)$ given by $g\mapsto (\tau_{B},e_{j})\cdot g$ is surjective.
\end{proof}
\end{proposition}
As a corollary of Proposition \ref{prp:2121}, it follows that the space of dual Galilean $p$-frames, $f^{*}_{p}(\mathcal{V},\tau,\gamma)$, is a $G_{p}$-torsor with respect to the (left) action
\begin{align}\label{eq:2125A}
f^{*}_{p}(\mathcal{V},\tau,\gamma)\times G_{p} &\longrightarrow f^{*}_{p}(\mathcal{V},\tau,\gamma)\\
\left(\left(\begin{array}{c}\tau^{A}  \\ e^{i} \end{array}\right),g\right) & \mapsto g\cdot\left(\begin{array}{c}\tau^{B}  \\ e^{j} \end{array}\right)={\left(\begin{array}{@{}c|c@{}}{{\Lambda}^{B}}_{A} & 0  \\\hline{v^{j}}_{A} & {{R}^{j}}_{i}  \\ \end{array}\right)}{\left(\begin{array}{c}\tau^{A}  \\ e^{i} \end{array}\right)}. \notag
\end{align}

Thus, the pair $(\tau,\gamma)$ defines the desired reduction of frames. However, the fact that $\gamma$ is only defined in $\text{Ker}(\tau)$ can be inconvenient in practice. Fortunately, one can provide an alternative but equivalent structure in which this problem is not present (a similar thing happens in the more familiar Leibnizian structures \cite{XBKM}). Indeed, given a $\tau$, prescribing a tensor $h\in \vee^{2}\mathcal{V}$ satisfying $h^{ab}\tau_{bc}=0$, with $\text{rank}(h)=d-(p+1)$ and Riemannian signature is equivalent to prescribing a $\gamma$ \footnote{To see this, take any basis $\{e_{i}\}$ of $\text{Ker}(\tau)$. The tensors $h\in \vee^{2}\text{Ker}(\tau)$ with $\text{rank}(h)=d-(p+1)$ and Riemannian signature can be uniquely written as $h=h^{ij}e_{i}\otimes e_{j}$ where $h^{ij}$ are real numbers defining a symmetric $\left(d-p-1\right)\times \left(d-p-1\right)$ matrix which is positive definite. It is clear that the space of such tensors and the space of Riemannian metrics in $\text{Ker}(\tau)$ are canonically bijective through the map $h^{ij}=(\gamma_{ij})^{-1}$ where $\gamma_{ij}=\gamma(e_{i},e_{j})$ (in this context, the word 'canonical' means that the bijection does not depend on the choice of basis). Finally, for any tensor $h^{ab}\in\vee^{2}\mathcal{V}$ one has $h^{ab}\tau_{bc}=0\Leftrightarrow h^{ab}\in\vee^{2}\text{Ker}(\tau)$ and the claim follows.}. In what follows, we will refer either to $(\tau,\gamma)$ or $(\tau,h)$ without loss of generality depending on which formulation is more convenient given the context.

The following notion of \textit{longitudinal frames} is crucial in the study of $p$NC connections.
\begin{definition}\label{def:2122}
\textit{Let $\tau\in \vee^{2}\mathcal{V}^{*}$ with $\text{rank}(\tau)=p+1$ and Lorentzian signature, and let $\gamma$ be a metric of Riemannian signature in $\text{Ker}(\tau)$. We define the space of longitudinal frames, denoted $Lf_{p}(\mathcal{V},\tau)$, as the space of ordered $(p+1)$-tuples of vectors that are orthonormal with respect to $\tau$, that is}
\begin{equation}\label{eq:2123A}
Lf_{p}(\mathcal{V},\tau):=\left\{\tau_{A}\in (\mathcal{V})^{p+1} \vert \,\,\,\,\tau(\tau_{A},\tau_{B})=\eta_{AB} \right\}.
\end{equation}
\end{definition}
From Proposition \ref{prp:2121}, it follows that the longitudinal group $LG_{p}$ acts on the space of longitudinal frames. Furthermore, such an action enjoys the property of regularity, as stated in the following proposition.
\begin{proposition}\label{prp:2122}
\textit{The space $Lf_{p}(\mathcal{V},\tau)$ is a $LG_{p}$-torsor with respect to the (right) action}
\begin{align}\label{eq:21211}
Lf_{p}(\mathcal{V},\tau)\times LG_{p} &\longrightarrow Lf_{p}(\mathcal{V},\tau),\\
\left(\tau_{A},\left({\Lambda^{A}}_{B},V_{B}\right)\right) & \mapsto \tau_{A}\cdot({\Lambda^{A}}_{B},V_{B})={\Lambda^{A}}_{B}\tau_{A}+{V}_{B}. \notag
\end{align}
\begin{proof}
As before, we have to check that for any $\tau_{A}\in Lf_{p}(\mathcal{V},\tau)$ the map $LG_{p}\longrightarrow Lf_{p}(\mathcal{V},\tau)$ given by $({\Lambda^{A}}_{B},V_{B})\mapsto{\Lambda^{A}}_{B}\tau_{A}+{V}_{B}$ is a bijection. That it is injective can be seen by acting with $\tau(\tau_{C},\cdot)$ on the equation ${\Lambda^{A}}_{B}\tau_{A}+{V}_{B}={\Lambda'^{A}}_{B}\tau_{A}+{V'}_{B}$, which, after contracting with $\eta^{CD}$, gives ${\Lambda^{D}}_{B}={\Lambda'^{D}}_{B}$ and, hence, $V_{B}=V'_{B}$. Surjectivity follows from Proposition \ref{prp:2121}: pick $\tau_{A}$ and any other $\tau'_{A}\in Lf_{p}(\mathcal{V},\tau)$ and complete each of them to construct two Galilean $p$-frames. From Proposition \ref{prp:2121}, these are going to be related by \eqref{eq:2125}, and in particular $\tau_{A}'={\Lambda^{B}}_{A}\tau_{B}+V_{A}$ for some ${\Lambda^{B}}_{A}\in O(1,p)$ and $V_{A}\in\text{Ker}(\tau)$.
\end{proof}
\end{proposition}
Finally, we introduce one last notion of frames which can be thought of as the 'dual' of longitudinal frames.
\begin{definition}\label{def:2123}
\textit{Let $\tau\in \vee^{2}\mathcal{V}^{*}$ with $\text{rank}(\tau)=p+1$ and Lorentzian signature, and let $\gamma$ be a metric of Riemannian signature in $\text{Ker}(\tau)$. We define the space of longitudinal co-frames, denoted $LCf_{p}(\mathcal{V},\tau)$, as the following space of ordered $(p+1)$-tuples of 1-forms} 
\begin{equation}\label{eq:21212}
LCf_{p}(\mathcal{V},\tau):=\left\{\tau^{A}\in (\mathcal{V}^{*})^{p+1} \vert \,\,\,\, \tau=\eta_{AB}\tau^{A}\otimes\tau^{B} \right\}.
\end{equation}
\end{definition}
We shall mention two facts about the space of longitudinal co-frames. First, there is a surjective map $Lf_{p}(\mathcal{V},\tau)\longrightarrow LCf_{p}(\mathcal{V},\tau)$ sending each longitudinal frame $\tau_{A}$ to the unique longitudinal co-frame $\tau^{A}$ satisfying $\tau^{B}(\tau_{A})={\delta^{B}}_{A}$\footnote{It is easy to check that for a given longitudinal frame $\tau_{A}$ there is a unique $(p+1)$-tuple of 1-forms satisfying $\tau=\eta_{AB}\tau^{A}\tau^{B}$ and $\tau^{B}(\tau_{A})={\delta^{B}}_{A}$.}. Nevertheless, this map is not injective as any two longitudinal frames related by a pure boost $\tau'_{A}=\tau_{A}+V_{A}$ will map to the same longitudinal co-frame. Hence, for a given longitudinal frame we can always use without loss of generality a 'dual' longitudinal co-frame but not conversely. Second, from \eqref{eq:2125A} it follows that $O(1,p)$ acts on $LCf_{p}(\mathcal{V},\tau)$ and, again, the action is regular.
\begin{proposition}\label{prp:2123}
\textit{The space $LCf_{p}(\mathcal{V},\tau)$ is a $O(1,p)$-torsor with respect to the (left) action}
\begin{align}\label{eq:21211}
LCf_{p}(\mathcal{V},\tau)\times O(1,p) &\longrightarrow LCf_{p}(\mathcal{V},\tau),\\
\left(\tau^{A},{\Lambda^{B}}_{A}\right) & \mapsto {{\Lambda}^{B}}_{A}\tau^{A}. \notag
\end{align}
\end{proposition}
The proof is analogue to that of Proposition \ref{prp:2122}.
\subsection{$p$-brane Newton--Cartan Structures}
In the previous section we introduced the metric pieces that at each spacetime point realise $G_{p}$. Here we extend such pieces to the spacetime manifold $M$. These define a $p$NC structure on $M$. New tensor fields that play the role of 'inverse metrics' are also introduced.
\begin{definition}\label{def:1}
\textit{A $p$-brane Newton--Cartan structure is a triplet $(M,\tau,h)$ consisting of a smooth manifold $M$, a smooth rank-$(p+1)$ symmetric tensor $\tau_{ab}$ of Lorentzian signature, and a smooth rank-$(d-p-1)$ symmetric tensor $h^{ab}$ of Riemannian signature, such that}
\begin{equation}\label{eq:39}
h^{ab}\tau_{bc}=0.
\end{equation}
\end{definition}
At each $q\in M$, the field $\tau$ defines a vector subspace $\text{Ker}(q,\tau)\subset T_{q}M$ consisting of the space of vectors $V_{q}\in T_{q}M$ satisfying $\tau(V_{q},\cdot)=0$. This, in turn, defines a distribution on $M$ given by $\text{Ker}(M,\tau):=\bigsqcup\limits_{q\in M}\text{Ker}(q,\tau)$. Sometimes, we will refer to this distribution as the \textit{transverse space}. In general, integrability of transverse space is not assumed, but we will see in the next section that so as to admit a compatible torsion-free connections a $p$NC structure must have an integrable $\text{Ker}(M,\tau)$ (see Proposition \ref{prp:4}). Equivalently, $p$NC structures can also be defined using a Riemannian metric $\gamma$ in $\text{Ker}(M,\tau)$ instead of $h$, as discussed in the previous section. Without loss of generality, we will refer either to $h$ or $\gamma$ depending on which formulation is more convenient given the context. 

A consequence of Proposition \ref{prp:2121} is that, using $(\tau,h)$, it is possible to define a reduction of the frame bundle based on the group $G_{p}$. 
\begin{definition}\label{def:2}
\textit{Let $(M,\tau,h)$ be a $p$NC structure. Let $q\in M$ and let $F(T_{q}M)$ be the space of frames of the tangent space at $q$. We define the space of Galilean $p$-frames at $q$ as the space}
\begin{equation}\label{eq:40}
f_{p}(q,\tau,h):=\left\{ (\tau_{A},e_{i})_{q}\in F(T_{q}M) \,\,\,\,\vert \,\,\,\, \{(e_{i})_{q}\}\in\text{Ker}(q,\tau), \,\,\tau((\tau_{A})_{q},(\tau_{B})_{q})=\eta_{AB}, \,\,\gamma((e_{i})_{q},(e_{j})_{q})=\delta_{ij}   \right\}.
\end{equation}
\textit{Then, the bundle of Galilean $p$-frames is defined as}
\begin{equation}\label{eq:41}
f_{p}(M,\tau,h):=\bigsqcup\limits_{q\in M} f_{p}(q,\tau,h).
\end{equation}
\end{definition}
From the work done in the previous section it follows that the bundle of Galilean $p$-frames is a principal $G_{p}$-bundle. In general, it has no smooth global sections but, for every $q$ in $M$, there is always an open neighbourhood $U\subset M$ in which smooth local sections exist. The space of such sections will be denoted $\Gamma(f_{p}(U,\tau,h))$ and referred to as the space of Galilean $p$-frames on $U$. 

In a way exactly analogous to the frame bundle of Galilean $p$-frames, one can define the bundle of longitudinal frames $Lf_{p}(M,\tau)$ and longitudinal co-frames $LCf_{p}(M,\tau)$. 
\begin{definition}\label{def:3}
\textit{Let $(M,\tau,h)$ be a $p$NC structure and let $q\in M$. We define the space of longitudinal frames at $q$ as the space}
\begin{equation}\label{eq:42}
Lf_{p}(q,\tau):=\left\{ (\tau_{A})_{q}\in (T_{q}M)^{p+1} \,\,\,\,\vert \,\,\,\,\tau((\tau_{A})_{q},(\tau_{B})_{q})=\eta_{AB}\right\}.
\end{equation}
\textit{Then, the bundle of longitudinal frames is defined as}
\begin{equation}\label{eq:43}
Lf_{p}(M,\tau):=\bigsqcup\limits_{q\in M} Lf_{p}(q,\tau).
\end{equation}
\end{definition}
\begin{definition}\label{def:4}
\textit{Let $(M,\tau,h)$ be a $p$NC structure and let $q\in M$. We define the space of longitudinal co-frames at $q$ as the space}
\begin{equation}\label{eq:44}
LCf_{p}(q,\tau):=\left\{ (\tau^{A})_{q}\in (T_{q}^{*}M)^{p+1} \,\,\,\,\vert \,\,\,\, \tau_{q}=\eta_{AB}(\tau^{A})_{q}(\tau^{B})_{q} \right\}.
\end{equation}
\textit{Then, the bundle of longitudinal co-frames is defined as}
\begin{equation}\label{eq:45}
LCf_{p}(M,\tau):=\bigsqcup\limits_{q\in M} LCf_{p}(q,\tau).
\end{equation}
\end{definition}

From Propositions \ref{prp:2122} and \ref{prp:2123} it follows that $Lf_{p}(M,\tau)$ and $LCf_{p}(M,\tau)$ are principal bundles of the groups $LG_{p}$ and $O(1,p)$, respectively. As for the Galilean $p$-frames, $\Gamma(Lf_{p}(U,\tau))$ denotes the space of local smooth sections of $Lf_{p}(U,\tau)$ on an open set $U$ in $M$, and will be referred to as the space of longitudinal frames on $U$. It will work analogously for the bundle of longitudinal co-frames. Let us stress here that, as discussed in the previous section, for each longitudinal frame $\tau_{A}\in Lf_{p}(U,\tau)$, there is a unique dual longitudinal co-frame $\tau^{A}\in LCf_{p}(U,\tau)$ (the converse, however, is not true). Therefore, given a $\tau_{A}$, a dual $\tau^{A}$ can be used without loss of generality. 

The results in the propositions of the previous section hold here for each fiber of the corresponding bundle, and can be extended in the obvious way to the local smooth sections. For the sake of clarity, we write explicitly the case of longitudinal frames. It works analogously for the rest of bundles.

\begin{proposition}\label{prp:1}
\textit{The space of longitudinal frames on $U\subset M$, $\Gamma(Lf_{p}(U,\tau))$, is the set of $(p+1)$-tuples of smooth vector fields in $U$, $\tau_{A}\in\left(\Gamma(TU)\right)^{p+1}$, satisfying}
\begin{equation}\label{eq:46}
\tau(\tau_{A},\tau_{B})=\eta_{AB},
\end{equation}
\textit{and it is a torsor of the group $C^{\infty}(U,LG_{p})$ of $C^{\infty}$ functions from $U$ to the longitudinal group $LG_{p}$, with (right) action}
\begin{align}\label{eq:47}
\Gamma(Lf_{p}(U,\tau))\times C^{\infty}(U,LG_{p}) &\longrightarrow \Gamma(Lf_{p}(U,\tau)),\\
\tau_{A},\left({\Lambda^{A}}_{B},V_{B}\right) & \mapsto \tau_{A}\cdot({\Lambda^{A}}_{B},V_{B})={\Lambda^{A}}_{B}\tau_{A}+{V}_{B}, \notag
\end{align}
\textit{where $\Lambda\in C^{\infty}(U,O(1,p))$ and $V_{A}\in \left(\Gamma(\text{Ker}(U,\tau))\right)^{p+1}$}.
\end{proposition}
The proof is the natural generalisation of that in Proposition \ref{prp:2122}. 

Finally, it will be useful to introduce a pair of tensor fields that play the role of inverse metrics of $\tau_{ab}$ and $h^{ab}$. These, however, are not unique, because their definition depends on the chosen way of projecting vectors $X_{q}\in T_{q}M$ into $\text{Ker}(q,\tau)$. 
\begin{definition}\label{def:5}
\textit{Let $\tau_{A}$ be a longitudinal frame. We define the projector associated to $\tau_{A}$ as the map $\overset{\tau}{P}:\Gamma(TM)\longrightarrow \Gamma(\text{Ker}(M,\tau))$ given by}
\begin{equation}\label{eq:48}
\overset{\tau}{P}(X)=X-\tau^{A}(X)\tau_{A},
\end{equation}
\textit{where $X\in\Gamma(TM)$ and $\tau^{A}$ is the dual longitudinal co-frame of $\tau_{A}$}\footnote{We recall that for each longitudinal frame $\tau_{A}$ there is a unique dual longitudinal co-frame $\tau^{A}$, as stated in the discussion above Proposition \ref{prp:2123}}.
\end{definition}
In what follows, the use of the superscript $\tau$ indicates that the quantity wearing it depends on the choice of longitudinal frame $\tau_{A}$, as it is the case of the projectors \eqref{eq:48}. In index notation, the projectors are given by
\begin{equation}\label{eq:48A}
\overset{\tau}{P}\tensor{\vphantom{P}}{^a_b}={\delta^{a}}_{b}-\tau_{A}^{a}\tau^{A}_{b}.
\end{equation}
With this we can introduce a notion of inverse of $\tau_{ab}$ and $h^{ab}$ as follows.
\begin{definition}\label{def:6}
\textit{Let $\tau_{A}$ be a longitudinal frame. We define $\overset{\tau}{\tau}\tensor{\vphantom{\tau}}{^a^b}$ as}
\begin{equation}\label{eq:49}
\overset{\tau}{\tau}\tensor{\vphantom{\tau}}{^a^b}:=\eta^{BC}\tau_{B}^{a}\tau_{C}^{b},
\end{equation}
\textit{and $\overset{\tau}{h}_{ab}$ as}
\begin{equation}\label{eq:50}
\overset{\tau}{h}(X,Y):=\gamma(\overset{\tau}{P}(X),\overset{\tau}{P}(Y)),
\end{equation}
\textit{where $X,Y\in \Gamma(TM)$.}
\end{definition}
These fields satisfy the usual orthogonality conditions
\begin{align}\label{eq:51}
&\overset{\tau}{\tau}\tensor{\vphantom{\tau}}{^a^b}\tau_{ab}=p+1,\\
&\overset{\tau}{\tau}\tensor{\vphantom{\tau}}{^a^b}\overset{\tau}{h}_{bc}=0,\\
&h^{ab}\overset{\tau}{h}_{bc}+\overset{\tau}{\tau}\tensor{\vphantom{\tau}}{^a^b}\tau_{bc}={\delta^{a}}_{c},
\end{align}
which can be easily checked by, for example, working in a Galilean $p$-frame. To conclude this section, in the following proposition we provide the transformation law of these new pieces when moving from one longitudinal frame to another.
\begin{proposition}\label{prp:2}
\textit{Let $\tau_{A},\tau'_{A}$ be two longitudinal frames. These are related by (see Proposition \ref{prp:1})}
\begin{equation}\label{eq:52}
\tau'_{A}={\Lambda^{B}}_{A}\tau_{B}+V_{A},
\end{equation}
\textit{with $\Lambda\in C^{\infty}(M,O(1,p))$ and $V_{A}\in \left(\Gamma(\text{Ker}(M,\tau))\right)^{p+1}$. Then, the associated projectors and inverse metric fields are related by}
\begin{align}\label{eq:53}
&\overset{\tau'}{P}\tensor{\vphantom{P}}{^a_b}= \overset{\tau}{P}\tensor{\vphantom{P}}{^a_b}-{(\Lambda^{-1})^{A}}_{B}\tau^{B}_{b}V^{a}_{A},\\
&\overset{\tau'}{\tau}\tensor{\vphantom{\tau}}{^a^b}=\overset{\tau}{\tau}\tensor{\vphantom{\tau}}{^a^b}+2\eta^{AB}{\Lambda^{C}}_{B}V^{(a}_{A}\tau^{b)}_{C}+\eta^{AB}V^{a}_{A}V^{b}_{B},\\
&\overset{\tau'}{h}_{ab}=\overset{\tau}{h}_{ab}-2{(\Lambda^{-1})^{A}}_{B}V^{c}_{A}\overset{\tau}{h}_{c(b}\tau^{B}_{a)}+{(\Lambda^{-1})^{A}}_{B}{(\Lambda^{-1})^{C}}_{D}\overset{\tau}{h}(V_{A},V_{C})\tau^{B}_{a}\tau^{D}_{b}.
\end{align}
\end{proposition}
Proving this proposition consists in just plugging \eqref{eq:52} into the definitions of the projector and inverse metrics given above.

As a last comment, it is worth noting that $p$NC structures interpolate between Leibnizian structures when $p$ is set to zero, and Lorentzian structures when $p=d-1$. While the latter case is obvious, it is interesting to discuss the former in more detail. First, notice that for $p=0$ the structure group is $G_{0}=Gal(d)$. Furthermore, the space of longitudinal co-frames is a torsor of the group composed of a single element (see Proposition \ref{prp:2123}). That is, there is a unique 1-form $\tau_{a}$ for which $\tau_{ab}=\tau_{a}\tau_{b}$ and, in addition, it satisfies $\tau_{a}h^{ab}=0$. Thus, the pair $(\tau_{a},h^{ab})$ forms a Leibnizian structure. More generally, when $p=0$ all the results of this section reduce manifestly to those of Leibnizian structures  \cite{XBKM,BernalSanchez}.

\section{$p$-brane Newton--Cartan Connections}

In order to describe gravity in non-relativistic regimes, the metric structure that realises the symmetry group is, in general, not enough, and it must be supplemented with a compatible connection. In this section we study the space of connections on $TM$ compatible with a given $p$NC structure\footnote{Whenever we make use of the term 'connection' in this work, we refer to a Koszul connection on $TM$. To avoid confusion, let us remind that a Koszul connection on $TM$ is a map $\nabla\colon\Gamma(TM)\to\text{End}\Gamma(TM)$ which is $C^{\infty}(M)$-linear, and such that for all $X\in\Gamma(TM)$ the endomorphism $\nabla_{X}$ satisfies the Leibniz rule, i.e.\, $\nabla_{X}(fY)=X(f)Y+f\nabla_{X}Y$ for all $f\in C^{\infty}(M)$ and $Y\in \Gamma(TM)$.}. Unlike in the case of relativistic structures, the conditions of compatibility with the tensor structure together with vanishing torsion do \textit{not} determine uniquely a connection. This fact is sometimes referred to as the equivalence problem in the literature \cite{XBKM}. Solving it consists in determining the additional data that has to be prescribed on the manifold (e.g.\, torsion) in order to fix uniquely a connection. In general, given a $p$NC structure, a connection which is torsion-free and compatible with the structure might be non-unique, or might not exist. 

In this section we first classify the $p$NC structures into Aristotelian and Augustinian (in analogy with the classification of Leibnizian structures). The former have a notion of absolute transverse space. The latter are the subclass of Aristotelian structures that admit torsion-free compatible connections. Focusing on the torsion-free case we present a solution of the equivalence problem for Augustinian $p$NC structures. This provides a class of connections that can be thought of as the analogue of the Levi--Civita connection of relativistic structures. These generalise the \textit{torsion-free special connections} of Leibnizian structures. The conventions have been chosen so that our results reduce to those in \cite{XBKM} when $p=0$ and are comparable to those in \cite{SNC} when $p=1$. 

\subsection{The Equivalence Problem in Non-Relativistic Structures}
Let $(M,\tau,h)$ be a $p$NC structure. We denote $\mathcal{D}(M,\tau,h)$ the space of connections on $TM$ compatible with $\tau$ and $h$, that is,
\begin{equation}
i)\,\,\,\nabla \tau=0, \,\,\,\,\,\,\,\,\,\,\, ii)\,\,\,\nabla h=0.
\end{equation}

Similarly, we denote $\mathcal{D}_{0}(M,\tau,h)$ the subspace of $\mathcal{D}(M,\tau,h)$ consisting of the connections with vanishing torsion (as discussed above, this space might be empty). Before focusing on the torsion-free case, here we shall give some general results about $\mathcal{D}(M,\tau,h)$. 

In non-relativistic structures there is no analogue of the Fundamental Theorem of Riemannian Geometry. Such lack can be understood as follows. Consider a spacetime $(M,g)$ where $M$ is a smooth manifold and $g_{ab}$ a metric of Lorentzian signature, and denote by $\mathcal{D}(M,g)$ the space of affine connections on $TM$ compatible with $g_{ab}$. It can be proven that the map 
\begin{align}\label{eq:54}
\mathcal{D}(M,g) &\longrightarrow \Gamma(\wedge^{2}T^{*}M\otimes TM),\\
\nabla & \mapsto \text{Tor}(\nabla), \notag
\end{align}
is a bijection. Hence, there exists a unique connection $\overset{g}{\nabla}$ in its kernel, i.e. there exists a unique connection that is torsion-free and compatible with $g_{ab}$. This is, of course, the Levi--Civita connection of $g_{ab}$. In addition, this shows that the condition of compatibility puts no constraints on the torsion. All this discussion is summarised in the following proposition for relativistic structures.
\begin{proposition}\label{prp:3} (See, e.g.\, \cite{XBKM}) \textit{The space $\mathcal{D}(M,g)$ of connections compatible with a Lorentzian structure $(M,g)$ possesses the structure of vector space, the origin of which is the Levi--Civita connection of $g$, $\overset{g}{\nabla}$, and $\mathcal{D}(M,g)$ is then isomorphic to $\Gamma(\wedge^{2}T^{*}M\otimes TM)$.} 
\end{proposition}
In non-relativistic structures such as $p$NC none of these results hold because the map analogue to \eqref{eq:54},
\begin{align}\label{eq:56}
\mathcal{D}(M,\tau,h) &\longrightarrow \Gamma(\wedge^{2}T^{*}M\otimes TM),\\
\nabla & \mapsto \text{Tor}(\nabla), \notag
\end{align}
is not a bijection. In general, it is nor injective neither surjective. On the one hand, this means that there might be none, or more than one compatible connections with zero torsion. On the other hand, unlike in the relativistic case, compatible connections do not have arbitrary torsion. 

Given a $p$NC structure, the constraints on the torsion of compatible connections are provided in the following proposition. 
\begin{proposition}\label{prp:4}
\textit{Let $(M,\tau,h)$ be a $p$-brane Newton--Cartan structure and let $\nabla$ be a connection in $\mathcal{D}(M,\tau,h)$, the space of connections compatible with $(\tau,h)$. Then, for all longitudinal frames $\tau_{A}$ the following equations hold}
\begin{align}\label{eq:57}
&i)\,\,\, \tau^{A}\left(T(V,W)\right)=\text{d}\tau^{A}(V,W),\\ 
&ii)\,\,\, \eta_{A(B}\tau^{A}\left(T(\tau_{C)},V)\right)=\eta_{A(B}\text{d}\tau^{A}\left(\tau_{C)},V\right),
\end{align} 
\textit{where $T$ is the torsion tensor of $\nabla$ and $V,W\in \Gamma\left(\text{Ker}(M,\tau)\right)$ are any pair of transverse vector fields.}
\end{proposition}

\begin{proof}
First, we shall prove that $i)$ and $ii)$ hold for one $\tau_{A}$, and after that we will show that if they hold for one, then they hold for all of them. For any affine connection $\nabla$ one has
\begin{equation}\label{eq:A58}
\tau^{A}\left(T(X,Y)\right)=\text{d}\tau^{A}(X,Y)-\left(\nabla_{X}\tau^{A}(Y)-\nabla_{Y}\tau^{A}(X)\right),
\end{equation}
for all $X,Y\in\Gamma(TM)$. Using that $\nabla \tau=0$ one has
\begin{equation}\label{eq:A59}
\nabla_{b}\tau^{A}_{a}=-\eta^{AD}\eta_{BC}\tau^{c}_{D}\tau^{C}_{a}\nabla_{b}\tau^{B}_{c},
\end{equation}
so that \eqref{eq:A58} becomes
\begin{equation}\label{eq:A60}
\tau^{A}\left(T(X,Y)\right)=\text{d}\tau^{A}(X,Y)+\eta^{AD}\eta_{BC}\left(\tau^{C}(Y)\nabla_{X}\tau^{B}(\tau_{D})-\tau^{C}(X)\nabla_{Y}\tau^{B}(\tau_{D})\right),
\end{equation} 
and $i)$ follows immediately. Evaluating \eqref{eq:A60} on $\tau_{E},V$, contracting with $\eta_{AF}$ and symmetrising one gets
\begin{equation}\label{eq:A61}
\eta_{A(F}\tau^{A}\left(T(\tau_{E)},V)\right)=\eta_{A(F}\text{d}\tau^{A}(\tau_{E)},V)-\eta_{B(E}\nabla_{V}\tau^{B}(\tau_{F)}),
\end{equation}
but for all $X\in\Gamma(TM)$ we have
\begin{equation}\label{eq:A62}
0=\nabla_{X}\tau(\tau_{E},\tau_{F})=2\eta_{B(E}\nabla_{X}\tau^{B}(\tau_{F)}),
\end{equation}
so \eqref{eq:A61} reduces to
\begin{equation}\label{eq:A63}
\eta_{A(F}\tau^{A}\left(T(\tau_{E)},V)\right)=\eta_{A(F}\text{d}\tau^{A}(\tau_{E)},V),
\end{equation}
which is $ii)$. Now we have to check that $i)$ and $ii)$ also hold for any other longitudinal frame $\tau'_{A}$. From the results in the previous sections, it follows that any two longitudinal frames $\tau'_{A}$ and $\tau_{A}$ and their corresponding dual co-frames are related by
\begin{align}\label{eq:A64}
\tau_{A}&={\Lambda^{B}}_{A}\tau'_{B}+V_{A},\\ \notag
\tau^{A}&={(\Lambda^{-1})^{A}}_{B}\tau'^{B},\notag
\end{align}
where $V_{A}$ is a transverse vector field and $\Lambda$ a Lorentz matrix. Since $i)$ holds for $\tau_{A}$, we have
\begin{align}\label{eq:A65}
{(\Lambda^{-1})^{A}}_{B}\tau'^{B}\left(T(V,W)\right)=\left(\text{d}{(\Lambda^{-1})^{A}}_{B}\right)\wedge\tau'^{B}(V,W)+{(\Lambda^{-1})^{A}}_{B}\text{d}\tau'^{B}(V,W), 
\end{align}
but the first term in the RHS vanishes and then $i)$ holds also for $\tau'_{A}$. Now we use this result in order to write $ii)$, that holds for $\tau_{A}$, as
\begin{equation}\label{eq:A66}
\eta_{A(B}{\Lambda^{E}}_{C)}{(\Lambda^{-1})^{A}}_{D}\tau'^{D}(T(\tau'_{E},V))={(\Lambda^{-1})^{A}}_{D}\eta_{A(B}{\Lambda^{E}}_{C)}\text{d}\tau'^{D}(\tau'_{E},V)-\eta_{A(B}{\Lambda^{D}}_{C)}\text{d}{(\Lambda^{-1})^{A}}_{D}(V),
\end{equation}
but
\begin{align}\label{eq:A67}
2\eta_{A(B}{\Lambda^{D}}_{C)}\text{d}{(\Lambda^{-1})^{A}}_{D}&=-(\eta_{AB}{(\Lambda^{-1})^{A}}_{D}\text{d}{\Lambda^{D}}_{C}+\eta_{AC}{(\Lambda^{-1})^{A}}_{D}\text{d}{\Lambda^{D}}_{B})\\ \notag
&=-(\eta_{AD}{\Lambda^{A}}_{B}\text{d}{\Lambda^{D}}_{C}+\eta_{AD}{\Lambda^{A}}_{C}\text{d}{\Lambda^{D}}_{B})\\ \notag
&=-\text{d}\left(\eta_{AD}{\Lambda^{A}}_{B}{\Lambda^{D}}_{C}\right)=-\text{d}\eta_{BC}=0,
\end{align}
and after some straightforward manipulation the surviving terms in \eqref{eq:A66} reduce to equation $ii)$ for $\tau'_{A}$.
\end{proof}

These conditions do not depend on the choice of longitudinal frame and, thus, they refer to the structure of the metrics $\tau$ and $h$. In fact, when studying $p$NC geometries as the leading terms of a covariant expansion of General Relativity \cite{TomasDavidLuca2,VdB,ObersExpansion}, it is useful to rewrite the results in Proposition \ref{prp:4} in terms of such metrics and in a general coordinate chart as follows
\begin{align}\label{eq:58}
\tau_{\rho\lambda}{T^{\lambda}}_{\mu\nu}h^{\mu\alpha}h^{\nu\beta}&=\partial_{[\mu}\tau_{\nu]\rho}h^{\mu\alpha}h^{\nu\beta},\\
\left(\overset{\tau}{\tau}\tensor{\vphantom{\tau}}{^\mu_\alpha}\tau_{\rho\lambda}+\overset{\tau}{\tau}\tensor{\vphantom{\tau}}{^\mu_\rho}\tau_{\alpha\lambda}\right){T^{\lambda}}_{\mu\nu}h^{\nu\beta}&=\left(\overset{\tau}{\tau}\tensor{\vphantom{\tau}}{^\mu_\alpha}\partial_{[\mu}\tau_{\nu]\rho}+\overset{\tau}{\tau}\tensor{\vphantom{\tau}}{^\mu_\rho}\partial_{[\mu}\tau_{\nu]\alpha}+\frac{1}{2}\partial_{\nu}\tau_{\alpha\rho}\right)h^{\nu\beta}.
\end{align}

It is natural to classify $p$-brane Newton--Cartan structures in analogy with the classification of Leibnizian structures. Let us recall that a Leibnizian structure with integrable transverse space (or, equivalently, $\tau\wedge\text{d}\tau=0$) is called \textit{Aristotelian}. If furthermore the absolute clock is closed, $\text{d}\tau=0$, we say the structure is \textit{Augustinian}. The following notions constitute a generalisation of Aristotelian and Augustinian structures.
\begin{definition}\label{def:7}
\textit{An Aristotelian $p$-brane Newton--Cartan structure is a $p$-brane Newton--Cartan structure, $(M,\tau,h)$, satisfying }
\begin{align}\label{eq:61}
\text{d}\tau^{A}(V,W)=0,
\end{align} 
\textit{for all transverse vector fields $V,W$, where $\tau_{A}$ and $\tau^{A}$ are a longitudinal frame and its dual.}
\end{definition}

Of course, by Frobenius' theorem (see e.g.\, \cite{Schutz}), Aristotelian $p$NC structures can be defined equivalently as $p$NC structures with integrable transverse space. 

\begin{definition}\label{def:8}
\textit{An Augustinian $p$-brane Newton--Cartan structure is a $p$-brane Newton--Cartan structure, $(M,\tau,h)$, satisfying }
\begin{align}\label{eq:63}
&i)\,\,\, \text{d}\tau^{A}(V,W)=0,\\ 
&ii)\,\,\, \eta_{A(B}\text{d}\tau^{A}\left(\tau_{C)},V\right)=0,\notag
\end{align} 
\textit{for all transverse vector fields $V,W$, where $\tau_{A}$ and $\tau^{A}$ are a longitudinal frame and its dual.}
\end{definition}

We notice that only the class of Augustinian $p$NC structures admit torsion-free compatible connections. Also, note that, for $p=0$, both Aristotelian and Augustinian structures reduce to those in the literature.

Finally, let us compare our compatibility conditions with those in the literature \cite{Malament,XBKM,BernalSanchez}. Given a $p$NC structure we require, in particular,
\begin{equation}\label{eq:65}
\nabla_{a}\tau_{b}\tau_{c}=0.
\end{equation}
In the Leibnizian case ($p=0$), this is more general than the usual compatibility condition in the literature (e.g.\, \cite{XBKM,BernalSanchez})
\begin{equation}\label{eq:66}
\nabla_{a}\tau_{b}=0.
\end{equation}
However, in the special case of $p$NC structures that admit torsion-free compatible connections (i.e.\,in the case of Augustinian structures), condition \eqref{eq:66} follows from \eqref{eq:65} \footnote{Indeed, $\nabla_{a}\tau_{b}\tau_{c}=0$ implies $\tau_{(b}\nabla_{\vert a\vert}\tau_{c)}=0$. Also, from Proposition \ref{prp:4}, torsion-freeness and compatibility of $\nabla$ imply $2\nabla_{[a}\tau_{b]}=0$. Combining these two equations and contracting with $\tau^{c}$, one has $\nabla_{a}\tau_{b}=-\tau^{c}\tau_{b}\nabla_{a}\tau_{c}=-\tau^{c}\tau_{b}\nabla_{c}\tau_{a}$. But $\nabla_{a}\tau_{b}$ is symmetric so  $\nabla_{a}\tau_{b}=\nabla_{(a}\tau_{b)}=-\tau^{c}\tau_{(b}\nabla_{\vert c\vert}\tau_{a)}=0$.}. Hence, both are equivalent and, indeed, we will see that our results for torsion-free connections compatible with Augustinian structures reduce, when $p=0$, to those in the literature.
\subsection{Torsion-free $p$NC Connections}
The space of connections on $TM$ is an affine space modelled on the space of tensors $\Gamma(TM\otimes T^{*}M\otimes T^{*}M)$. Consider an Augustinian $p$NC structure, $(M,\tau,h)$, and denote $\mathcal{D}_{0}(M,\tau,h)$ the space of torsion-free connections compatible with $\tau_{ab}$ and $h^{ab}$. Then, $\mathcal{D}_{0}(M,\tau,h)$ is an affine space modelled on a vector subspace of  $\Gamma(TM\otimes T^{*}M\otimes T^{*}M)$, according to the following proposition.
\begin{proposition}\label{prp:5}
\textit{The space $\mathcal{D}_{0}(M,\tau,h)$ of torsion-free affine connections compatible with an Augustinian $p$NC structure $(M,\tau,h)$ is an affine space modelled on the vector space}
\begin{equation}\label{eq:67}
\mathfrak{V}(M,\tau,h)=\left\{ {S^{a}}_{bc}\in\Gamma(TM\otimes T^{*}M\otimes T^{*}M)\,\,\,\vert\,\,\, {S^{a}}_{[bc]}=0,\,\,\, {S^{d}}_{a(b}\tau_{c)d}=0, \,\,\, h^{d(c}{S^{b)}}_{ad}=0\right\}.
\end{equation}
\end{proposition}
The action of $\mathfrak{V}(M,\tau,h)$ on $\mathcal{D}_{0}(M,\tau,h)$ is
\begin{align}\label{eq:68}
\mathcal{D}_{0}(M,\tau,h)\times \mathfrak{V}(M,\tau,h)&\longrightarrow \mathcal{D}_{0}(M,\tau,h)\\
({\Gamma^{\lambda}}_{\mu\nu},{S^{\lambda}}_{\mu\nu}) & \mapsto  {\Gamma^{\lambda}}_{\mu\nu}+{S^{\lambda}}_{\mu\nu}. \notag
\end{align}
An affine space does not have the structure of vector space because it lacks a notion of zero. Solving the equivalence problem reduces to determining an explicit origin for the affine space $\mathcal{D}_{0}(M,\tau,h)$. In order to do so, we are going to use the following lemma.
\begin{lemma}\label{lmm:1}
\textit{Let}
\begin{itemize}
\item \textit{$\mathcal{D}$ be an affine space modelled on a vector space $\mathfrak{V}$. The substraction map for any two elements $\nabla,\nabla'$ in $\mathcal{D}$ is denoted $\nabla-\nabla'$.}
\item \textit{$\mathfrak{W}$ be a vector space isomorphic to $\mathfrak{V}$. We denote the isomorphism $\varphi \colon \mathfrak{V}\longrightarrow \mathfrak{W}$.}
\item \textit{$\Theta$ be an affine map modelled on $\varphi$, i.e.\,a map $\Theta \colon \mathcal{D}\longrightarrow \mathfrak{W}$ satisfying $\Theta(\nabla)-\Theta(\nabla')=\varphi(\nabla-\nabla')$ for all $\nabla,\nabla'\in \mathcal{D}$.}
\end{itemize}
\textit{Then, $\Theta$ is a bijection.}
\end{lemma}
Hence, if there is a vector space $\mathfrak{W}$ isomorphic to $\mathfrak{V}(M,\tau,h)$ and can construct an affine map $\Theta\colon \mathcal{D}_{0}(M,\tau,h)\longrightarrow \mathfrak{W}$ modelled on the isomorphism, then an origin for $\mathcal{D}_{0}(M,\tau,h)$ is uniquely determined as $\overset{0}{\nabla}=\text{Ker}\,\Theta$. If the construction of such map involves only the pieces in the structure, the origin $\overset{0}\nabla$ is canonical. This is the case of relativistic structures, where $\Theta(\nabla)=\text{Tor}(\nabla)$. However, for Leibnizian structures a field of observers $N$ must be prescribed in order to construct an affine map $\overset{N}{\Theta}$. Then, the origin $\overset{N}{\nabla}=\text{Ker}\,\overset{N}{\Theta}$ is not canonical and the field of observers $N$ is the additional structure required to fix uniquely a compatible connection \cite{XBKM}. In what follows, we are going to see that an analogous thing happens for Augustinian $p$NC structures: given a longitudinal frame (that can be thought of as the generalisation of a field of observers in Leibnizian structures) it is possible to determine uniquely a connection in $\mathcal{D}_{0}(M,\tau,h)$.

We shall begin by defining a parametrisation of $\mathfrak{V}(M,\tau,h)$. Given a longitudinal frame, $\tau_{A}$, we define the map 
\begin{align}\label{eq:69}
\overset{\tau}{\varphi}\colon \underset{p+1}{\oplus}{\Omega}^{2}(M)&\longrightarrow \mathfrak{V}(M,\tau,h),\\
{F_{A}}_{ab}&\mapsto \tau^{A}_{(b}{F}_{\vert A\vert c)d}h^{ad}, \notag
\end{align}
where $\underset{p+1}{\oplus}{\Omega}^{2}(M)$ denotes the direct sum of $p+1$ copies of the space of 2-forms $\Omega^{2}(M)$. Its elements are denoted $F_{A}$, $A$ being the index that labels each 2-form. This parametrisation has two significant advantages. First, for $p=0$ it reduces exactly to the parametrisation used in the literature to study Leibnizian structures in standard Newton--Cartan gravity \cite{XBKM}. Second, it is convenient to compare our results to those in \cite{SNC}.

The map $\overset{\tau}{\varphi}$ is linear, but it is not an isomorphism. Indeed, just by dimensional counting, one has $\text{dim}\left(\mathfrak{V}(M,\tau,h)\right)=(p+1)\left(d(d-p-1)/2\right)$, while $\text{dim}\left(\underset{p+1}{\oplus}{\Omega}^{2}(M)\right)=(p+1)\left(d(d-1)/2\right)$ (of course, we refer to the dimension of each fiber). In particular, for $p\geq0$ we have $\text{dim}\left(\underset{p+1}{\oplus}{\Omega}^{2}(M)\right)\geq\text{dim}\left(\mathfrak{V}(M,\tau,h)\right)$, and the latter inequality saturates when the former does. It follows that the kernel of $\overset{\tau}{\varphi}$ is not empty and it can be determined as shown in the following proposition.  
\begin{proposition}\label{prp:6}
\textit{Let $(M,\tau,h)$ be an Augustinian $p$NC structure, and let $\tau_{A}$ be a longitudinal frame. The kernel of $\overset{\tau}{\varphi}$ is the vector space $\overset{\tau}{\mathcal{K}}$ consisting of the polyforms $\overset{\tau}{K}_{A}\in \underset{p+1}{\oplus}\Omega^{2}(M)$ satisfying }
\begin{equation}\label{eq:73}
i)\,\,\, \overset{\tau}{K}_{A}(\tau_{B},V)=-\overset{\tau}{K}_{B}(\tau_{A},V),\,\,\,\,\,\,\,\,\, ii)\,\,\, \overset{\tau}{K}_{A}(V,W)=0,
\end{equation} 
\textit{for all transverse vector fields $V,W\in\Gamma(\text{Ker}(M,\tau))$.}
\end{proposition}
\begin{proof}
First, contracting equation
\begin{equation}\label{eq:AA1}
\tau^{A}_{(b}\overset{\tau}{K}_{\vert A\vert c)d}h^{ad}=0,
\end{equation}
with $\tau_{B}^{b}$ and $\overset{\tau}{h}_{ae}$ and re-organising terms conveniently, one has that the elements $\overset{\tau}{K}_{A}\in \text{Ker}\overset{\tau}{\varphi}$ are those satisfying
\begin{equation}\label{eq:AA2}
\overset{\tau}{K}_{A}(X,Y)=\overset{\tau}{K}_{A}(X,\tau_{B})\tau^{B}(Y)-\overset{\tau}{K}_{B}(\tau_{A},Y)\tau^{B}(X)+\overset{\tau}{K}_{B}(\tau_{A},\tau_{C})\tau^{B}(X)\tau^{C}(Y).
\end{equation}
Now we have to check that equations \eqref{eq:73} are equivalent to the condition $\overset{\tau}{K}_{A}\in\text{Ker}\overset{\tau}{\varphi}$. It is clear from \eqref{eq:AA2} that if $\overset{\tau}{K}_{A}\in\text{Ker}\overset{\tau}{\varphi}$ then equations \eqref{eq:73} are satisfied. At the same time, assuming that the pair of conditions in \eqref{eq:73} hold, it is easy to check, by working on a Galilean $p$-frame, that equation \eqref{eq:AA1} holds. Thus, $\overset{\tau}{K}_{A}\in\text{Ker}\overset{\tau}{\varphi}$.
\end{proof}

We shall view the vector space $\overset{\tau}{\mathcal{K}}$ as a 'gauge group' acting on $\underset{p+1}{\oplus}{\Omega}^{2}(M)$ as
\begin{align}\label{eq:add1}
\underset{p+1}{\oplus}{\Omega}^{2}(M)\times\overset{\tau}{\mathcal{K}}&\longrightarrow \underset{p+1}{\oplus}{\Omega}^{2}(M),\\
({F}_{A},\overset{\tau}{K}_{A})&\mapsto {F}_{A}+\overset{\tau}{K}_{A}. \notag
\end{align} 
Any two ${F}_{A},{F'}_{A}\in\underset{p+1}{\oplus}{\Omega}^{2}(M)$ related by a $\overset{\tau}{\mathcal{K}}$-transformation \eqref{eq:add1} map via \eqref{eq:69} to the same element. This gauge ambiguity can be removed by considering the space of $\overset{\tau}{\mathcal{K}}$-orbits in $\underset{p+1}{\oplus}{\Omega}^{2}(M)$. In fact, we are going to see that such a space is isomorphic to $\mathfrak{V}(M,\tau,h)$.

\begin{definition}\label{def:9}
\textit{Let $(M,\tau,h)$ be an Augustinian $p$NC structure, and let $\tau_{A}$ be a longitudinal frame. A $\overset{\tau}{\mathcal{K}}$-orbit in $\underset{p+1}{\oplus}{\Omega}^{2}(M)$ is dubbed a gravitational field strength with respect to $\tau_{A}$. The vector space of gravitational field strengths with respect to $\tau_{A}$ is denoted}
\begin{equation}\label{eq:74}
\overset{\tau}{\mathcal{F}}:=\underset{p+1}{\oplus}{\Omega}^{2}(M)/\overset{\tau}{\mathcal{K}}.
\end{equation}
\end{definition}

Notice that, indeed, $\overset{\tau}{\mathcal{F}}$ has the structure of vector space, as it is nothing but the quotient of a vector space by a vector subspace. We shall denote $[F_{A}]_{\tau}$ the elements in $\overset{\tau}{\mathcal{F}}$, i.e.\,$[F_{A}]_{\tau}$ is the set of elements in $\underset{p+1}{\oplus}\Omega^{2}(M)$ that belong to the $\overset{\tau}{\mathcal{K}}$-orbit through $F_{A}$. Conversely, we will use the notation $\overset{\tau}{F}_{A}$ to denote a generic representative of $[F_{A}]_{\tau}$ in $\underset{p+1}{\oplus}\Omega^{2}(M)$. The terminology used in this definition will be justified later on. 

\begin{proposition}\label{prp:7}
\textit{Let $(M,\tau,h)$ be an Augustinian $p$NC structure and let $\tau_{A}$ be a longitudinal frame. The space $\overset{\tau}{\mathcal{F}}$ of gravitational field strengths with respect to $\tau_{A}$ is isomorphic to the space $\mathfrak{V}(M,\tau,h)$.}
\end{proposition}
\begin{proof}
Consider the linear map 
\begin{align}\label{eq:75}
\overset{\tau}{\overline{\varphi}}\colon \overset{\tau}{\mathcal{F}}&\longrightarrow \mathfrak{V}(M,\tau,h),\\
[{F}_{A}]_{\tau}&\mapsto \tau^{A}_{(b}\overset{\tau}{F}_{\vert A\vert c)d}h^{ad}, \notag
\end{align}
where $\overset{\tau}{F}_{A}$ is any representative of $[{F}_{A}]_{\tau}\in \overset{\tau}{\mathcal{F}}$ (by construction, there is no dependence on the representative chosen). We are going to show that $\overset{\tau}{\overline{\varphi}}$ admits an inverse, given by
\begin{align}\label{eq:76}
\overset{\tau}{\overline{\varphi}}\tensor{\vphantom{\overline{\varphi}}}{^{-1}}\colon \mathfrak{V}(M,\tau,h)&\longrightarrow  \overset{\tau}{\mathcal{F}},\\
{S^{a}}_{bc}&\mapsto  [{F}_{Aab}]_{\tau}=[2\overset{\tau}{h}_{c[b}{S^{c}}_{a]d}\tau^{d}_{A}]_{\tau}. \notag
\end{align}
Indeed, if 
\begin{equation}\label{eq:E77}
{S^{a}}_{bc}=\tau^{A}_{(b}\overset{\tau}{F}_{\vert A\vert c)d}h^{ad},
\end{equation}
then
\begin{equation}\label{eq:E78}
2\overset{\tau}{h}_{c[b}{S^{c}}_{a]d}\tau^{d}_{A}=\overset{\tau}{F}_{Aab}+\overset{\tau}{K}_{Aab},
\end{equation}
where 
\begin{equation}\label{eq:E79}
\overset{\tau}{K}_{Aab}=\tau^{c}_{A}\overset{\tau}{F}_{Bc[b}\tau^{B}_{a]}-\tau^{B}_{(a}\overset{\tau}{F}_{\vert B\vert c)e}\tau^{c}_{A}\tau^{e}_{C}\tau^{C}_{b}+\tau^{B}_{(b}\overset{\tau}{F}_{\vert B\vert c)e}\tau^{c}_{A}\tau^{e}_{C}\tau^{C}_{a}.
\end{equation}
However, it is easy to check that
\begin{equation}\label{eq:E80}
\overset{\tau}{K}_{A}(V,W)=0,\,\,\,\,\,\, \overset{\tau}{K}_{A}(V,\tau_{B})=\frac{1}{2}(\overset{\tau}{F}_{B}(\tau_{A},V)-\overset{\tau}{F}_{A}(\tau_{B},V))=-\overset{\tau}{K}_{B}(V,\tau_{A}),
\end{equation}
for all transverse vector fields $V,W$. Thus, $\overset{\tau}{K}_{A}\in\text{Ker}\,\overset{\tau}{\varphi}$ so $[\overset{\tau}{K}_{A}]_{\tau}=0$. This allows us to write
\begin{equation}\label{eq:E81}
\overset{\tau}{\overline{\varphi}}\tensor{\vphantom{\overline{\varphi}}}{^{-1}}\circ\overset{\tau}{\overline{\varphi}}\left([F_{A}]_{\tau}\right)=\overset{\tau}{\overline{\varphi}}\tensor{\vphantom{\overline{\varphi}}}{^{-1}}\left(\tau^{A}_{(b}\overset{\tau}{F}_{\vert A\vert c)d}h^{ad}\right)=[\overset{\tau}{F}_{A}+\overset{\tau}{K}_{A}]_{\tau}=[{F}_{A}]_{\tau}.
\end{equation}
\end{proof}
This, together with Proposition \ref{prp:5}, allow us to characterise $\mathcal{D}_{0}(M,\tau,h)$ as an affine space modelled on $\overset{\tau}{\mathcal{F}}$. For clarity, we shall state it as a proposition.

\begin{proposition}\label{prp:8}
\textit{Let $(M,\tau,h)$ be an Augustinian $p$NC structure, and let $\tau_{A}$ be a longitudinal frame. Then, the space $\mathcal{D}_{0}(M,\tau,h)$ of torsion-free connections compatible with the Augustinian $p$NC structure is an affine space modelled on the vector space $\overset{\tau}{\mathcal{F}}$ of gravitational field strenghts with respect to $\tau_{A}$.} 
\end{proposition}
The action of $\overset{\tau}{\mathcal{F}}$ on $\mathcal{D}_{0}(M,\tau,h)$ is given by
\begin{align}\label{eq:77}
\mathcal{D}_{0}(M,\tau,h)\times \overset{\tau}{\mathcal{F}}&\longrightarrow \mathcal{D}_{0}(M,\tau,h),\\
\left({\Gamma^{\lambda}}_{\mu\nu},[{F}_{A\mu\nu}]_{\tau}\right)&\mapsto {\Gamma^{\lambda}}_{\mu\nu}+\tau^{A}_{(\mu}\overset{\tau}{F}_{\vert A\vert \nu)\rho}h^{\lambda\rho}. \notag
\end{align}
We argued above that an origin for $\mathcal{D}_{0}(M,\tau,h)$ can be obtained by constructing a suitable affine map (see Lemma \ref{lmm:1}). A natural choice of such a map follows from the fact that
\begin{equation}\label{eq:78}
\overset{\tau}{\overline{\varphi}}\tensor{\vphantom{\overline{\varphi}}}{^{-1}}(\nabla-\nabla')_{Aab}=[2\overset{\tau}{h}_{c[b}{\nabla}_{a]}\tau^{c}_{A}-2\overset{\tau}{h}_{c[b}{\nabla'}_{a]}\tau^{c}_{A}]_{\tau}=[2\overset{\tau}{h}_{c[b}{\nabla}_{a]}\tau^{c}_{A}]_{\tau}-[2\overset{\tau}{h}_{c[b}{\nabla'}_{a]}\tau^{c}_{A}]_{\tau}.
\end{equation}

\begin{definition}\label{def:10}
\textit{Let $(M,\tau,h)$ be an Augustinian $p$NC structure, and let $\tau_{A}$ be a longitudinal frame. We define the map}
\begin{align}\label{eq:80}
\overset{\tau}\Theta \colon \mathcal{D}_{0}(M,\tau,h)&\longrightarrow \overset{\tau}{\mathcal{F}},\\
\nabla&\mapsto [\overset{\tau}{F}(\nabla)_{A}]_{\tau}, \notag
\end{align}
\textit{where}
\begin{equation}\label{eq:fieldstreght}
\overset{\tau}{F}(\nabla)_{Aab}:=2\overset{\tau}{h}_{c[b}{\nabla}_{a]}\tau^{c}_{A}.
\end{equation}
\textit{We refer to $[\overset{\tau}{F}(\nabla)_{A}]_{\tau}$ as the gravitational field strength induced by $\nabla$ with respect to $\tau_{A}$.}  
\end{definition}
From Lemma \ref{lmm:1} it follows that $\overset{\tau}\Theta$ is a bijection and therefore we can provide $\mathcal{D}_{0}(M,\tau,h)$ with an origin consisting of $\overset{\tau}{\nabla}=\text{Ker}\,{\overset{\tau}\Theta}$. This result solves the equivalence problem for Augustinian $p$NC structures.
\begin{proposition}\label{prp:9}
\textit{Let $\mathcal{D}_{0}(M,\tau,h)$ be the space of torsion-free connections compatible with an Augustinian $p$NC structure $(M,\tau,h)$. Given a longitudinal frame $\tau_{A}$, there exists a unique torsion-free, compatible connection $\overset{\tau}{\nabla}\in\mathcal{D}_{0}(M,\tau,h)$ such that the gravitational field strength induced by $\overset{\tau}\nabla$ with respect to $\tau_{A}$ vanishes. We refer to $\overset{\tau}{\nabla}$ as \textit{the torsion-free special connection associated to $\tau_{A}$}. Furthermore, the connection components of $\overset{\tau}{\nabla}$ in a general coordinate chart are given by}
\begin{equation}\label{eq:81}
\overset{\tau}{\Gamma}\tensor{\vphantom{\Gamma}}{^\lambda_\mu_\nu} =\frac{1}{2}\overset{\tau}{\tau}\tensor{\vphantom{\tau}}{^\lambda^\gamma}\left(\partial_{\mu}\tau_{\nu\gamma}+\partial_{\nu}\tau_{\mu\gamma}-\partial_{\gamma}\tau_{\mu\nu}\right)+\frac{1}{2}h^{\lambda\gamma}\left(\partial_{\mu}\overset{\tau}{h}_{\nu\gamma}+\partial_{\nu}\overset{\tau}{h}_{\mu\gamma}-\partial_{\gamma}\overset{\tau}{h}_{\mu\nu}\right).
\end{equation}
\end{proposition}
\begin{proof}
Uniqueness is already proven, so it only remains to be shown that in a general coordinate chart the connection components are those in \eqref{eq:81}. Let $X,Y,Z$ be any triplet of vector fields. Since $[\overset{\tau}{F}(\overset{\tau}\nabla)_{A}]_{\tau}=0$ we know from the proof of Proposition \ref{prp:6} that $\overset{\tau}{F}(\overset{\tau}\nabla)_{A}$ satisfies
\begin{equation}\label{eq:EE1}
\overset{\tau}{F}(\overset{\tau}\nabla)_{A}(X,Y)=\overset{\tau}{F}(\overset{\tau}\nabla)_{A}(X,\tau_{B})\tau^{B}(Y)-\overset{\tau}{F}(\overset{\tau}\nabla)_{B}(\tau_{A},Y)\tau^{B}(X)+\overset{\tau}{F}(\overset{\tau}\nabla)_{B}(\tau_{A},\tau_{C})\tau^{B}(X)\tau^{C}(Y).
\end{equation}
At the same time, from the definition of $\overset{\tau}{F}(\overset{\tau}\nabla)$ in \eqref{eq:fieldstreght}, we have
\begin{equation}\label{eq:EE2}
\overset{\tau}{F}(\overset{\tau}\nabla)_{A}(X,Y)=\overset{\tau}{h}(\overset{\tau}{\nabla}_{X}\tau_{A},Y)-\overset{\tau}{h}(\overset{\tau}{\nabla}_{Y}\tau_{A},X).  
\end{equation}
Combining these two equations and using compatibility and torsion-freeness of $\overset{\tau}{\nabla}$, one gets after some manipulations the couple of equations
\begin{align}\label{eq:EE3}
2\overset{\tau}{h}(\overset{\tau}{\nabla}_{X}Y,Z)=&X[\overset{\tau}{h}(Y,Z)]+Y[\overset{\tau}{h}(Z,X)]-Z[\overset{\tau}{h}(X,Y)]\\ \notag
&+\overset{\tau}{h}([X,Y],Z)+\overset{\tau}{h}([Z,X],Y)-\overset{\tau}{h}([Y,Z],X)\\ \notag
&+2\tau^{A}(Z)\left[\overset{\tau}{h}(Y,\overset{\tau}{\nabla}_{X}\tau_{A})-\tau^{B}(Y)\overset{\tau}{h}(X,\overset{\tau}{\nabla}_{\tau_{A}}\tau_{B})\right], \notag
\end{align}
and
\begin{align}\label{eq:EE4}
2\tau(\overset{\tau}{\nabla}_{X}Y,Z)=&X[\tau(Y,Z)]+Y[\tau(Z,X)]-Z[\tau(X,Y)]\\ \notag
&+\tau([X,Y],Z)+\tau([Z,X],Y)-\tau([Y,Z],X). \notag
\end{align}
Evaluating the second one in a coordinate basis
\begin{equation}\label{eq:EE5}
X=\partial_{\alpha},\,\,\,\,\,\, Y=\partial_{\beta},\,\,\,\,\,\, Z=\partial_{\gamma},
\end{equation}
and contracting with $\overset{\tau}{\tau}\tensor{\vphantom{\tau}}{^\gamma^\rho}$, one gets
\begin{equation}\label{eq:EE6}
\overset{\tau}{\Gamma}\tensor{\vphantom{\Gamma}}{^\rho_\beta_\alpha}=\frac{1}{2}\overset{\tau}{\tau}\tensor{\vphantom{\tau}}{^\gamma^\rho}\left(\partial_{\alpha}\tau_{\beta\gamma}+\partial_{\beta}\tau_{\alpha\gamma}-\partial_{\gamma}\tau_{\alpha\beta}\right)+h^{\rho\gamma}\overset{\tau}{h}_{\gamma\mu}  \overset{\tau}{\Gamma}\tensor{\vphantom{\Gamma}}{^\mu_\beta_\alpha}.
\end{equation}
Evaluating \eqref{eq:EE3} in the basis \eqref{eq:EE5} and contracting with $h^{\gamma\rho}$, gives
\begin{equation}\label{eq:EE7}
h^{\rho\gamma}\overset{\tau}{h}_{\gamma\mu} \overset{\tau}{\Gamma}\tensor{\vphantom{\Gamma}}{^\mu_\beta_\alpha}=\frac{1}{2}h^{\rho\gamma}\left(\partial_{\alpha}\overset{\tau}{h}_{\beta\gamma}+\partial_{\beta}\overset{\tau}{h}_{\alpha\gamma}-\partial_{\gamma}\overset{\tau}{h}_{\alpha\beta}\right).
\end{equation}
Plugging \eqref{eq:EE7} into \eqref{eq:EE6}, the result follows.
\end{proof}

Our torsion-free special connections reduce, when $p=0$, to the torsion-free special connections of Leibnizian structures \cite{Torfreespecial1,Torfreespecial2} and, thus, can be thought of as their generalisation. Given a longitudinal frame $\tau_{A}$, we can take the associated $\overset{\tau}{\nabla}$ as the origin of $\mathcal{D}_{0}(M,\tau,h)$ and, consequently, it acquires the structure of vector space. In this sense, one can think of $\overset{\tau}{\nabla}$ as the analogue of the Levi--Civita connection of relativistic structures. Actually, for $p=d-1$, $\overset{\tau}{\nabla}$ is nothing but the Levi--Civita connection of $\tau_{ab}$.

For the sake of clarity, we summarise all the results provided so far in the form of a theorem, as follows.

\begin{theorem}\label{th:1}
\textit{Let $\mathcal{D}_{0}(M,\tau,h)$ be the space of torsion-free connections compatible with an Augustinian $p$NC structure $(M,\tau,h)$. Given a longitudinal frame $\tau_{A}$, $\mathcal{D}_{0}(M,\tau,h)$ has the structure of vector space, the origin of which is the torsion-free special connection $\overset{\tau}{\nabla}$, and $\mathcal{D}_{0}(M,\tau,h)$ is then naturally isomorphic to the vector space $\overset{\tau}{\mathcal{F}}$ of gravitational field strengths with respect to $\tau_{A}$.}
\end{theorem} 

In practice, this theorem shows that all connections in $\mathcal{D}_{0}(M,\tau,h)$ can be given explicitly as
\begin{equation}\label{eq:82}
{\Gamma^{\lambda}}_{\mu\nu}=\overset{\tau}{\Gamma}\tensor{\vphantom{\Gamma}}{^\lambda_\mu_\nu}+\tau^{A}_{(\mu}\overset{\tau}{F}_{\vert A\vert \nu)\rho}h^{\lambda\rho}.
\end{equation}
However, this depends on a choice of longitudinal frame $\tau_{A}$. In general, two pairs $(\tau_{A},[{F}_{A}]_{\tau})$ and $(\tau'_{A},[{F'}_{A}]_{\tau'})$ might map to the same connection via \eqref{eq:82}. Their relation is given in the following proposition.

 \begin{proposition}\label{prp:10}
 \textit{Consider two pairs $(\tau_{A},[{F}_{A}]_{\tau})$ and $(\tau'_{A},[{F'}_{A}]_{\tau'})$. $\tau'_{A}$ and $\tau_{A}$ (and the respective longitudinal co-frames) are related by (cf.\,Proposition \ref{prp:1}) }
\begin{align}\label{eq:83}
\tau'_{A}={\Lambda^{B}}_{A}\tau_{B}+V_{A},\,\,\,\,\,\,\tau'^{A}={(\Lambda^{-1})^{A}}_{B}\tau^{B}.
\end{align}
 \textit{Then, if $(\tau_{A},[{F}_{A}]_{\tau})$ and $(\tau'_{A},[{F'}_{A}]_{\tau'})$ map to the same connection, i.e.}
\begin{equation}\label{eq:84}
{\Gamma^{\lambda}}_{\mu\nu}=\overset{\tau}{\Gamma}\tensor{\vphantom{\Gamma}}{^\lambda_\mu_\nu}+\tau^{A}_{(\mu}\overset{\tau}{F}_{\vert A\vert \nu)\rho}h^{\lambda\rho}=\overset{\tau'}{\Gamma}\tensor{\vphantom{\Gamma}}{^\lambda_\mu_\nu}+\tau'^{A}_{(\mu}\overset{\tau'}{F'}_{\vert A\vert \nu)\rho}h^{\lambda\rho},
\end{equation}
\textit{one has}
\begin{align}\label{eq:85}
[F'_{A}]_{\tau'}=&{\Lambda^{B}}_{A}\left[\overset{\tau}{F_{B}}+\text{d}\overset{\tau}{h}\left({(\Lambda^{-1})^{C}}_{B}V_{C}\right)+\frac{1}{2}\tau^{D}\wedge\text{d}\left(\overset{\tau}{h}\left({(\Lambda^{-1})^{E}}_{D}V_{E},{(\Lambda^{-1})^{F}}_{B}V_{F}\right)\right)\right]_{\tau'}\\ \notag
&+{\Lambda^{B}}_{A}\left[{\Delta^{D}}_{B}\wedge\overset{\tau}{h}\left({(\Lambda^{-1})^{E}}_{D}V_{E}\right)\right]_{\tau'}\\ \notag
&+{\Lambda^{B}}_{A}\left[\left(\frac{1}{2}\eta^{DE}\eta_{CF}\text{d}\tau^{F}(\tau_{B},\tau_{E})\right)\overset{\tau}{h}\left({(\Lambda^{-1})^{G}}_{D}V_{G}\right)\wedge\tau^{C}-\overset{\tau}{h}\left({(\Lambda^{-1})^{E}}_{C}V_{E},{(\Lambda^{-1})^{F}}_{D}V_{F}\right){\Delta^{D}}_{B}\wedge\tau^{C}\right]_{\tau'}, \notag
\end{align}
\textit{where ${\Delta^{D}}_{B}$ is the 1-form}
\begin{equation}\label{eq:86}
{\Delta^{D}}_{B}=\frac{1}{2}\left(\eta^{DE}\eta_{BF}\text{d}\tau^{F}(\tau_{E})-\text{d}\tau^{D}(\tau_{B})\right).
\end{equation}
\end{proposition}

\begin{proof}
Noticing that 
\begin{equation}\label{eq:A10.6}
\overset{\tau}{\overline{\varphi}}\left(\left[F_{A}\right]_{\tau}\right)=\overset{\tau'}{\overline{\varphi}}\left(\left[{\Lambda^{B}}_{A}\overset{\tau}{F}_{B}\right]_{\tau'}\right),
\end{equation}
it is easy to show that
\begin{equation}\label{eq:A10.8}
[F'_{A}]_{\tau'}=\left[{\Lambda^{B}}_{A}\overset{\tau}{F}_{B}+\overset{\tau'}{F}_{A}(\overset{\tau}{\nabla})\right]_{\tau'},
\end{equation}
and we recall that $\overset{\tau'}{F}_{A}(\overset{\tau}{\nabla})$ is given by \eqref{eq:fieldstreght},
\begin{equation}\label{eq:A10.3}
\overset{\tau'}{F}_{Aab}(\overset{\tau}{\nabla})=2\overset{\tau'}{h}_{c[b}{\overset{\tau}{\nabla}}_{a]}\tau'^{c}_{A}.
\end{equation}

After a long computation using the transformation laws in \eqref{eq:53}, one gets that $\overset{\tau'}{F}_{A}(\overset{\tau}{\nabla})$ is
\begin{align}\label{eq:A10.9}
\overset{\tau'}{F}_{Aab}(\overset{\tau}{\nabla})=& {\Lambda^{B}}_{A}\left(\overset{\tau}{F}_{B}(\overset{\tau}{\nabla})_{ab}+4\delta^{C}_{[B}\delta^{F}_{D]}{(\Lambda^{-1})^{E}}_{C}\overset{\tau}{h}\left(V_{E}\right)_{c}\tau^{D}_{[b}\overset{\tau}{\nabla}_{a]}\tau^{c}_{F}\right)\\ \notag
&+{\Lambda^{B}}_{A}\left({(\Lambda^{-1})^{E}}_{[D}{(\Lambda^{-1})^{F}}_{B]}2\tau^{D}_{[a}\overset{\tau}{\nabla}_{b]}\overset{\tau}{h}\left(V_{E},V_{F}\right)\right)\\ \notag
&+{\Lambda^{B}}_{A}\left(\tau^{D}_{[a}\overset{\tau}{h}([{(\Lambda^{-1})^{E}}_{D}V_{E},{(\Lambda^{-1})^{F}}_{B}V_{F}])_{b]}\right)\\ \notag
&+{\Lambda^{B}}_{A}\left(\eta^{DE}\eta_{CF}\text{d}\tau^{F}(\tau_{B},\tau_{E})\overset{\tau}{h}\left({(\Lambda^{-1})^{G}}_{H}V_{G},{(\Lambda^{-1})^{J}}_{D}V_{J}\right)\tau^{C}_{[a}\tau^{H}_{b]}\right)\\ \notag
&+{\Lambda^{B}}_{A}\left(\left(\text{d}\overset{\tau}{h}({(\Lambda^{-1})^{C}}_{B}V_{C})\right)_{ab}+\left(\frac{1}{2}\tau^{D}\wedge\text{d}\overset{\tau}{h}\left({(\Lambda^{-1})^{E}}_{D}V_{E},{(\Lambda^{-1})^{F}}_{B}V_{F}\right)\right)_{ab}\right)\\ \notag
&+{\Lambda^{B}}_{A}\left({\Delta^{D}}_{B}\wedge\overset{\tau}{h}({(\Lambda^{-1})^{E}}_{D}V_{E})\right)_{ab}\\ \notag
&+{\Lambda^{B}}_{A}\left(\left(\frac{1}{2}\eta^{DE}\eta_{CF}\text{d}\tau^{F}(\tau_{B},\tau_{E})\right)\left(\overset{\tau}{h}\left({(\Lambda^{-1})^{G}}_{D}V_{G}\right)\wedge\tau^{C}\right)_{ab}\right)\\ \notag
&-{\Lambda^{B}}_{A}\left(\overset{\tau}{h}\left({(\Lambda^{-1})^{E}}_{C}V_{E},{(\Lambda^{-1})^{F}}_{D}V_{F}\right)\left({\Delta^{D}}_{B}\wedge\tau^{C}\right)_{ab}\right). \notag
\end{align}
However, taking into account \eqref{eq:A10.6}, it is not difficult to check that the first four lines in \eqref{eq:A10.9} belong to the kernel of $\overset{\tau'}{\varphi}$ (as they satisfy the equations \eqref{eq:73}) and, hence, they map to zero when taking the quotient $[\cdot]_{\tau'}$. Plugging \eqref{eq:A10.9} into \eqref{eq:A10.8}, the result follows.
\end{proof}

Although, for arbitrary $p$, equation \eqref{eq:85} is rather involved, for $p=0$ we have that $\text{d}\tau=0={\Delta^{D}}_{B}$, so it reduces manifestly to the transformation law of field strengths in standard Augustinian structures \cite{XBKM}. 

More formally, it is convenient to regard the pairs $(\tau_{A},[{F}_{A}]_{\tau})$ as living in the space $\bigsqcup_{\tau_{A}}\overset{\tau}{\mathcal{F}}$, the disjoint union of the family of sets $\{\overset{\tau}{\mathcal{F}}\}$ indexed by the longitudinal frames $\tau_{A}$ \footnote{Notice that, in the case $p=0$, the space $\overset{\tau}{\mathcal{F}}$ is nothing but $\Omega^{2}(M)$ and, hence, it does not depend on $\tau_{A}$. Thus, $\bigsqcup_{\tau_{A}}\overset{\tau}{\mathcal{F}}=\Gamma(Lf_{0}(M,\tau))\times\Omega^{2}(M)$, or, in the terminology of \cite{XBKM}, $FO(M,\tau)\times\Omega^{2}(M)$.}. Then, from Proposition \ref{prp:10}, it follows that the longitudinal group $LG_{p}$ acts on $\bigsqcup_{\tau_{A}}\overset{\tau}{\mathcal{F}}$ as
\begin{align}\label{eq:87}
\bigsqcup_{\tau_{A}}\overset{\tau}{\mathcal{F}}\times C^{\infty}(M,LG_{p}) &\longrightarrow \bigsqcup_{\tau_{A}}\overset{\tau}{\mathcal{F}},\\
\Big(\left(\tau_{A},[F_{A}]_{\tau}\right),(\Lambda,V)\Big) & \mapsto \left(\tau'_{A},[F'_{A}]_{\tau'}\right), \notag
\end{align}
where
\begin{equation}\label{eq:88}
\tau'_{A}={\Lambda^{B}}_{A}\tau_{B}+{V}_{A},
\end{equation}
and $[F'_{A}]_{\tau'}$ is given by \eqref{eq:85}. 

\begin{definition}\label{def:11}
\textit{Let $(M,\tau,h)$ be an Augustinian $p$NC structure. An orbit of $C^{\infty}(M,LG_{p})$ in $\bigsqcup_{\tau_{A}}\overset{\tau}{\mathcal{F}}$ is dubbed a gravitational field strength. The space of gravitational field strengths is denoted}
\begin{equation}\label{eq:89}
\mathfrak{F}(M,\tau,h):=\left(\bigsqcup_{\tau_{A}}\overset{\tau}{\mathcal{F}}\right)/C^{\infty}(M,LG_{p}).
\end{equation}
\end{definition}
We shall denote $\left[\tau_{A},[F_{A}]_{\tau}\right]$ the elements in $\mathfrak{F}(M,\tau,h)$. Notice that the space of field strengths $\mathfrak{F}(M,\tau,h)$ is canonical, as it only depends on the pieces of the Augustinian $p$NC structure. Furthermore, from Proposition \ref{prp:10}, it follows that the orbit of $C^{\infty}(M,LG_{p})$ through some $\left(\tau_{A},[F_{A}]_{\tau}\right)\in\bigsqcup_{\tau_{A}}\overset{\tau}{\mathcal{F}}$ is precisely the set of elements that map via \eqref{eq:82} to the same connection as $\left(\tau_{A},[F_{A}]_{\tau}\right)$. We conclude that the map
\begin{align}\label{eq:90}
 \mathfrak{F}(M,\tau,h)&\longrightarrow\mathcal{D}_{0}(M,\tau,h),\\
\left[\tau_{A},[F_{A}]_{\tau}\right] & \mapsto \overset{\tau}{\Gamma}\tensor{\vphantom{\Gamma}}{^\lambda_\mu_\nu}+\tau^{A}_{(\mu}\overset{\tau}{F}_{\vert A\vert \nu)\rho}h^{\lambda\rho}, \notag
\end{align}
where $(\tau_{A},\overset{\tau}F_{A})$ is any representative of $\left[\tau_{A},[F_{A}]_{\tau}\right]\in\mathfrak{F}(M,\tau,h)$, is a canonical bijection. Finally, we notice that, for $p=0$, our space of gravitational field strengths reduces to that of Newton--Cartan gravity \cite{XBKM}. The standard physical interpretation of the latter justifies the terminology used throughout this section. 
\section{Discussion}
Our results can be summarised as follows.

\begin{itemize}

\item The notion of $p$-brane Newton--Cartan structure has been recovered as a set of singular metrics that realise the group $G_{p}$ (the group of symmetries in tangent space of a background consisting of flat space containing a non-relativistic $p$-brane). A formal description of such structures has been provided for arbitrary $p$. Leibnizian and relativistic structures are recovered by setting $p=0$ and $p=d-1$, respectively.

\item Focusing on Augustinian $p$NC structures (the only $p$NC structures admitting torsion-free, compatible connections), a generalisation of the torsion-free, special connections of Newton--Cartan gravity \cite{Torfreespecial1,Torfreespecial2} has been obtained. This solves the equivalence problem in the Augustinian case. Finally, the corresponding space $\mathcal{D}_{0}(M,\tau,h)$ of torsion-free, compatible connections has been determined exactly, as summarised in Theorem \ref{th:1}.

\end{itemize}

It is interesting to contrast our results with those in the original formulation of SNC geometry \cite{SNC}. First, we notice that the set of conventional curvature constraints of \cite{SNC} imply, among other things, that the corresponding $p$NC structure is Augustinian. Indeed, equations (3.26) in \cite{SNC} (which correspond to projections of the curvature constraints) are nothing but \eqref{eq:63}. Then, imposing a Vielbein postulate a torsion-free connection on $TM$ compatible with the $p$NC structure is introduced. After some manipulation, the authors managed to put it in the form \eqref{eq:82}, in agreement with our results. The remaining conventional curvature constraints fix the field strengths $\overset{\tau}{F}_{A\mu\nu}$ as
\begin{equation}\label{eq:91}
\overset{\tau}{F}_{A\mu\nu}=2D_{[\mu}{m_{\nu]}}^{A},
\end{equation}
where ${m_{\mu}}^{A}$ is a gauge field and $D_{\mu}$ the gauge covariant derivative. Also, it was noticed that $\overset{\tau}{F}_{A\mu\nu}$ has certain ambiguity corresponding to the shift
\begin{equation}\label{eq:92}
\overset{\tau}{F}_{A\mu\nu}\mapsto \overset{\tau}{F}_{A\mu\nu} + Y_{ABC}\left(\tau^{B}\wedge\tau^{C}\right)_{\mu\nu},
\end{equation}
for arbitrary parameters $Y_{ABC}$. In addition, the authors were able to provide an interpretation of this ambiguity in terms of the gauge transformations of the theory. However, we have shown that this is not the only ambiguity in $\overset{\tau}{F}_{A\mu\nu}$. Recalling Proposition \ref{prp:6}, one has that the most general 2-forms $\overset{\tau}{K}_{A\mu\nu}$ by which $\overset{\tau}{F}_{A\mu\nu}$ can be shifted not only have longitudinal components but also mixed ones. While, as noticed in \cite{SNC}, the longitudinal components are arbitrary, the mixed ones must satisfy
\begin{equation}\label{eq:93}
\overset{\tau}{K}_{A}(\tau_{B},V)=-\overset{\tau}{K}_{B}(\tau_{A},V).
\end{equation}
It would be interesting to investigate in future work whether this further 'ambiguity in the ambiguity' can be explained by the gauge transformations and choice of conventional curvature constraints of \cite{SNC}.

Continuations of this work can be suggested from several perspectives. Motivated by classical work about Newton--Cartan structures (see \cite{NullRed} or more recently \cite{Morand}), it would be interesting to study if $p$NC geometry can result from a null dimensional reduction of a relativistic ambient space. Alternatively, Double Field Theory has proved being powerfull for classifying and describing non-Riemannian geometries \cite{DFT}, and it seems likely that $p$NC structures are included in such classification. Dually to $p$NC structures, a surge of interest for Carrollian geometry, see e.g.\,\cite{NewtonCarroll}, motivates the question of whether Carrollian structures and connections \cite{XBKM2,HartongCarroll} can be generalised to the case of $p$-branes with Carroll symmetry \cite{Biel}. Finally, it would also be interesting to provide an intrinsic definition of a subset in $\mathcal{D}_{0}(M,\tau,h)$ that corresponds precisely to the connections satisfying \eqref{eq:91}. This would be analogous to the case of Leibnizian structures, where the subset of Galilean connections satisfying the Duval-Künzle condition \cite{Torfreespecial1,DuvalKunzle} are precisely those whose field strength is closed.
\vspace{1.5cm}

\large\textsc{Acknowledgements}.
I would like to thank Tomas Ortín and Luca Romano for useful discussion and a careful reading of the manuscript. I am also grateful to Ángel Murcia, Carlos S. Shahbazi, Joaquim Gomis and Eric A. Bergshoeff for fruitful conversations. Finally, I specially thank Xavier Bekaert and Kevin Morand for helpful and attentive correspondence. This work has been supported in part by the MINECO/FEDER, UE grant PGC2018-095205-B-I00, and by the ``Centro de Excelencia Severo Ochoa'' Program grant  SEV-2016-0597. I am further supported by a ``Centro de Excelencia Internacional UAM/CSIC'' FPI pre-doctoral grant.


\end{document}